\let\latexaddtocontents\addtocontents
\documentclass[a4paper]{quantumarticle}
\let\addtocontents\latexaddtocontents
\pdfoutput=1
\usepackage[breaklinks]{hyperref}
\usepackage{amsthm, amssymb}
\usepackage[sorting=none, style=numeric-comp]{biblatex}
\usepackage{caption}
\usepackage{subcaption}
\usepackage{tikz}
\usetikzlibrary{quantikz}
\usetikzlibrary{positioning}
\usepackage[shortlabels]{enumitem}
\usepackage{soul}
\usepackage[capitalise]{cleveref}

\renewcommand{\H}{\ensuremath{H}}
\newcommand{\RZ}{\ensuremath{R_Z}}
\newcommand{\CU}{\ensuremath{CR_Z}}
\newcommand{\Z}{\ensuremath{Z}}
\newcommand{\CZ}{\ensuremath{CZ}}
\newcommand{\CX}{\ensuremath{CX}}

\newcommand{\dEbit}{|[draw, circle, thin, minimum size=0.5em]| {}}
\newcommand{\dGate}[1]{|[draw, thin, minimum size=1.1em]| {\text{\small $#1$}} \qw}
\newcommand{\dPhase}[1]{|[draw, circle, thin, minimum size=1.1em]| {\text{\small $#1$}}}
\newcommand{\dRz}[1]{\dPhase{#1} \qw}
\newcommand{\dH}{|[draw, thin, minimum size=0.7em]| \qw}

\newcommand{\brac}[1]{\left( #1 \right)}
\newcommand{\bigO}[1]{\mathcal{O}\!\left( #1 \right)}

\newcommand{\ie}{\textit{i.e\@.}}
\newcommand{\eg}{\textit{e.g\@.}}
\newcommand{\repo}{\url{https://github.com/CQCL/pytket-dqc}}
\newcommand{\gateset}{$\left\{ \H, \RZ, \CU \right\}$}
\newcommand{\KaHyPar}{\texttt{KaHyPar}}
\newcommand{\pytketdqc}{\textsf{pytket\_dqc}}

\newcounter{nthm} 
\theoremstyle{definition}
\newtheorem{definition}[nthm]{Definition}
\newtheorem{remark}[nthm]{Remark}

\newtheorem{lemma}[nthm]{Lemma}
\theoremstyle{plain}
\newtheorem{theorem}[nthm]{Theorem}

\usepackage{algorithmicx}
\usepackage{algorithm}
\usepackage[noend]{algpseudocode}

\usepackage{tabularx}
\usepackage{booktabs}

\newcolumntype{b}{X}
\newcolumntype{m}{>{\hsize=.6\hsize}X}
\newcolumntype{s}{>{\hsize=.3\hsize}X}

\newcommand{\inout}[2]{
	\vspace{7pt}
	\hspace*{\algorithmicindent} \textbf{Input:} #1 \\
	\hspace*{\algorithmicindent} \textbf{Output:} #2

	\hrulefill
	}

\newcommand{\ce}{\texttt{Embed}}
\newcommand{\ces}{\texttt{EmbedSteiner}}

\newcommand{\cesd}{\texttt{EmbedSteinerDetach}}
\newcommand{\hp}{\texttt{Partition}}
\newcommand{\ph}{\texttt{PartitionHetero}}
\newcommand{\pe}{\texttt{PartitionEmbed}}
\newcommand{\phe}{\texttt{PartitionHeteroEmbed}}
\newcommand{\anneal}{\texttt{Annealing}}
\newcommand{\gss}{\texttt{FullG*-Simple}}
\newcommand{\gslp}{\texttt{FullG*-LP}}

\begin{document}

\title{Distributing circuits over heterogeneous, modular quantum computing network architectures}

\author[1,7]{Pablo Andres-Martinez}
\author[2,7]{Tim Forrer}
\author[1,7]{Daniel Mills}
\author[3]{Jun-Yi Wu}
\author[1]{Luciana Henaut}
\author[1]{Kentaro Yamamoto}
\author[2,5]{Mio Murao}
\author[1,4,6]{Ross Duncan}

\affil[1]{Quantinuum, Terrington House, 13-15 Hills Road, Cambridge CB2 1NL, UK}
\affil[2]{Department of Physics, Graduate School of Science, The University of Tokyo, Hongo 7-3-1, Bunkyo-ku, Tokyo 113-0033, Japan}
\affil[3]{Department of Physics, Tamkang University, 151 Yingzhuan Rd}
\affil[4]{Department of Physics and Astronomy, University College London, Gower Street, London, WC1E 6BT, UK}
\affil[5]{Trans-scale Quantum Science Institute, The University of Tokyo, Bunkyo-ku, Tokyo 113-0033, Japan}
\affil[6]{Department of Computer and Information Sciences, University of Strathclyde, 26 Richmond Street, Glasgow G1 1XH, UK}
\affil[7]{These authors contributed equally. Corresponding authors: \{daniel.mills, pablo.andresmartinez\}@quantinuum.com}

\maketitle

\begin{abstract}
    We consider a heterogeneous network of quantum computing modules, sparsely connected via Bell states. Operations across these connections constitute a computational bottleneck and they are likely to add more noise to the computation than operations performed within a module. We introduce several techniques for transforming a given quantum circuit into one implementable on a network of the aforementioned type, minimising the number of Bell states required to do so.
    
    We extend previous works on circuit distribution over fully connected networks to the case of heterogeneous networks. On the one hand, we extend the hypergraph approach of [Andres-Martinez \& Heunen. 2019] to arbitrary network topologies. We additionally make use of Steiner trees to find efficient realisations of the entanglement sharing within the network, reusing already established connections as often as possible. On the other hand, we extend the embedding techniques of [Wu, \textit{et al.} 2022] to networks with more than two modules. Furthermore, we discuss how these two seemingly incompatible approaches can be made to cooperate. Our proposal is implemented and benchmarked; the results confirming that, when orchestrated, the two approaches complement each other's weaknesses.

\end{abstract}

\newpage
\tableofcontents
\newpage

\section{Introduction}

Quantum computing providers are racing to scale up their systems, targeting qubit numbers and gate fidelities that would allow for demonstrations of quantum advantage on practical applications. As architectures scale up, their basic components grow farther apart, increasing the cost of communicating between them. Moreover, operations between distant components require more intermediary elements to be involved, thus making it challenging to maintain high fidelity as errors accumulate.

Distributed quantum computing \cite{caleffi2022distributed} provides an alternative: once a quantum computing module pushing the limits of current classical technology is engineered, it may be more practical to produce copies of it and connect them together than to produce larger singular devices. Indeed, researchers in academia and industry have proposed both short and long-term distributed quantum computing projects.
\begin{description}
    \item[Short-term.] An emerging field of research studies the use of classical postprocessing to `knit together' multiple quantum circuits~\cite{CutQC, Piveteau2022CircuitKnitting}, with the goal of simulating circuits that are too large to be run in current quantum computers. The quantum circuit is `cut' at different points, creating smaller subcircuits that can be run on current quantum computers. The classical postprocessing may be done `offline' --- \ie{} after the quantum computation has finished --- but the overhead scales exponentially with the number of cuts, so the technique is only applicable to circuits that can be split using few cuts. Practical applications can be found in the field of quantum chemistry, where knowledge of the symmetries of the system being modelled can be exploited to generate circuits in which two groups of qubits barely interact with each other~\cite{ChemKnitting}.
    \item[Long-term.] There is a history of academics proposing modular quantum computers~\cite{Rodney_2009, Monroe_2014, Nigmatullin_2016} and related technologies appear in the field of quantum internet~\cite{Wehner2018QuantumInternetReview, Cacciapuoti_2020}. In such modular architectures, it is expected that different modules will interact with each other throughout the computation via entanglement sharing. Currently, the challenge of high-rate generation of entangled states between different modules is too great for the technology to become widely applicable, but we can expect to eventually reach an inflexion point where the communication cost within a large enough module will be comparable to that of entanglement generation between separate modules~\cite{Rodney_2009}. The current road-map of IBM promises the release of the first prototype of a modular quantum computer (Heron) by the end of 2023, and Quantinuum plans to develop a modular quantum computer for its H5 generation.\footnote{The road-map of these companies is publicly available at their respective web-pages at the time of writing.} 
\end{description}

When we reach the inflexion point where modular architectures become advantageous, communication of quantum information between modules will be a significant bottleneck of the computation. It is thus essential to develop circuit optimisation methods that minimise the amount of quantum communication required to distribute a circuit. This is the purpose of the present manuscript. The methods we discuss here are also applicable to the short-term applications of classically simulated circuit knitting~\cite{Piveteau2022CircuitKnitting, ChemKnitting}, since reducing the amount of communication between modules is equivalent to reducing the number of cuts and, hence, the exponential classical overhead.

In this work, we assume that all quantum communication is carried out by the consumption of Bell pairs shared between modules. Previous works on distributed quantum computing focus either on the minimisation of the circuit's depth \cite{Cuomo_2023} or attempt to minimise the number of Bell pairs consumed~\cite{Andr_s_Mart_nez_2019, gsundaram_et_al, Sundaram2022GeneralDistribution, Junyi2022}. This manuscript falls into the second category, since we identify Bell pair generation and sharing as the main bottleneck of the computation. Among the works in this category, \cite{Andr_s_Mart_nez_2019, gsundaram_et_al, Junyi2022} assume a fully connected network of modules. \cite{Sundaram2022GeneralDistribution} studies heterogeneous networks, where not every pair of modules are connected to each other directly, and where each module may have different qubit register capacities.

The task of circuit distribution has some similarities with the qubit routing problem~\cite{cowtan_et_al, childs_et_al}, in that both are concerned with gate scheduling and assignment qubits to hardware registers. The main distinction between them lies in that the goal of routing is to implement a circuit on a \textit{single} module (with limited connectivity), whereas the distribution problem deals with the interaction between multiple modules. Thus, the distribution problem can be studied at a higher level of abstraction, where we may assume operations within a module to be comparatively free. This leads to distribution being naturally related to the mathematical problem of graph partitioning, whereas qubit routing is an instance of token swapping~\cite{childs_et_al}. Moreover, this distinction leads to a desirable separation of concerns: once a circuit is distributed, the next step on a compilation stack is to solve the routing problem for each of its subcircuits, optimising its implementation for the specific hardware constraints of the module it is assigned to.

In \cref{sec:DQC_problem} we give a precise definition of the circuit distribution problem. We review the relevant literature in \cref{sec:background}, focusing on approaches that minimise the number of Bell pairs consumed~\cite{Andr_s_Mart_nez_2019, gsundaram_et_al, Sundaram2022GeneralDistribution, Junyi2022}. The main contribution of our work is the generalisation of the approaches of \cite{Andr_s_Mart_nez_2019, Junyi2022} to target heterogeneous networks. We identify the key optimisation opportunities exploited by these generalisations and describe an approach to combine them in \cref{sec:solutions}. The proposed approach has been implemented as an open source project, \pytketdqc{}, which we benchmark in \cref{sec:benchmarks}.

\section{The DQC problem}
\label{sec:DQC_problem}

In this work we focus on the problem of distributing quantum circuits (DQC) over general networks of quantum computers, minimising the number of entangled resources required to do so. A network is comprised of a collection of quantum computers that we refer to as \emph{modules}. These modules are connected via quantum communication channels, with Local Operations and Classical Communication (LOCC) also available.
A quantum communication channel may be used to generate maximally entangled bipartite states between two modules.
We refer to a such shared state as an \emph{ebit}, and take it to be a Bell state:
\begin{equation}
    \frac{1}{\sqrt{2}} \brac{\ket{00} + \ket{11}}.
\end{equation}

Formally, the network is specified by an undirected graph $G = (V,E)$. Each vertex $\texttt{A} \in V$ corresponds to a module and each edge $(\texttt{A},\texttt{B}) \in E$ indicates that ebits may be prepared and shared between modules $\texttt{A}$ and $\texttt{B}$. Each module $\texttt{A} \in V$ is capable of managing $\omega(\texttt{A})$ qubits dedicated to computation --- its \textit{computation register} --- and $\epsilon(\texttt{A})$ qubits dedicated to communication --- its \textit{link qubit register}. Thus, $\epsilon(\texttt{A})$ determines the maximum number of connections that can be simultaneously maintained by module $\texttt{A}$. These link qubits are disentangled from the rest of the computation at the end of the communication protocol described in \cref{sec:EJPP}. Consequently, we may reuse the space in the link qubit register throughout the computation in order to establish new communications channels at different points in time.\footnote{In this work we abstract away details about inter-module entanglement generation and management. We refer the reader to \cite{Illiano_2022,Kozlowski_2020,Pirker_2019,Van_Meter_2022} for details and reviews of methods of constructing a complete 'quantum internet protocol stack'.}

We assume each of the modules is capable of universal quantum computation and we consider no restrictions on the module's internal qubit connectivity. The particular universal gate set, and the actual internal connectivity of the modules, may be accounted for by a later stage of circuit compilation \cite{Sivarajah_2020} acting individually on the local subcircuit assigned to each module. Our objective is to minimise the total number of ebits consumed, whose preparation and sharing is expected to be the bottleneck of any distributed quantum computation. Throughout the paper we consider that LOCC are comparatively free and assume that circuits are constructed using the gateset \gateset{}, which we depict using the following shorthand:
{
\centering
\begin{tikzpicture}
    \node (H) {
    \begin{tikzpicture}
        \node (gate) {
            \begin{quantikz}[column sep=2mm]
                \qw & \dH{} & \qw
            \end{quantikz}
        };
        \node[right=0mm of gate] (eq) {$=$};
        \node[right=0mm of eq, scale=0.9] {$\frac{1}{\sqrt{2}}\begin{pmatrix} 1&1 \\ 1&-1 \end{pmatrix}$,};
    \end{tikzpicture}
    };
    \node[right=2mm of H] (Rz) {
    \begin{tikzpicture}
        \node (gate) {
            \begin{quantikz}[column sep=2mm]
                \qw & \dRz{\alpha} & \qw
            \end{quantikz}
        };
        \node[right=0mm of gate] (eq) {$=$};
        \node[right=0mm of eq, scale=0.9] {$\begin{pmatrix} 1&0 \\ 0&e^{i \alpha} \end{pmatrix}$,};
    \end{tikzpicture}
    };
    \node[below right=3mm and -20mm of H] (CRz) {
    \begin{tikzpicture}
        \node (gate) {
            \begin{quantikz}[column sep=1mm]
                \qw & \ctrl{1} & \qw \\[-7pt]
                & \dPhase{\alpha} & \\[-7pt]
                \qw & \ctrl{-1} & \qw
            \end{quantikz}
        };
        \node[right=0mm of gate] (eq) {$=$};
        \node[right=0mm of eq, scale=0.9] {$\begin{pmatrix} 1&0&0&0 \\ 0&1&0&0 \\ 0&0&1&0 \\ 0&0&0&e^{i \alpha} \end{pmatrix}$.};
    \end{tikzpicture}
    };
\end{tikzpicture}
}
Evidently, \CZ{} and \Z{} gates are contained within this gateset, since they are particular instances of \CU{} gates and \RZ{} gates whose phase is $\pi$. 

The DQC problem in the case of such a network can be divided into two subproblems:
\begin{itemize}
    \item \textbf{Qubit allocation}. We must determine the allocation of each qubit of the circuit to a module $\texttt{A} \in V$ in the network.  The number of qubits allocated to each module $\texttt{A} \in V$ must not exceed the module's computation register size $\omega(\texttt{A})$. 
\end{itemize}
Once the circuit's qubits are allocated, some two-qubit gates may act on qubits allocated to different modules; we call these \textit{non-local} gates. We are interested in qubit allocations that reduce the number of ebits required to implement these non-local gates.
\begin{itemize}
    \item \textbf{Non-local gate distribution}. Once a qubit allocation is chosen, we must find a way to implement the non-local gates that arise. This may be done by consuming ebits and using LOCC. As in previous works~\cite{Andr_s_Mart_nez_2019,gsundaram_et_al,Sundaram2022GeneralDistribution,Junyi2022}, here we focus on the use of simultaneous gate teleportation --- which we refer to as the EJPP protocol --- as described in \cref{sec:EJPP}. 
\end{itemize}
The maximum number of EJPP protocols sharing data with module $\texttt{A} \in V$ at a given point in time should not exceed the module's link qubit register size of $\epsilon(\texttt{A})$ qubits. Approaches to enforce this capacity constraint have been explored in previous works~\cite{Sundaram2022GeneralDistribution,Junyi2022}. In this work, we present techniques that assume $\epsilon(\texttt{A})$ is not bounded and, hence, may allow an arbitrary number of simultaneous quantum channels. Such an assumption is unreasonable in practice and in \cref{sec:limited links} we discuss a simple algorithm that, given an already distributed circuit, modifies it so that the bound to $\epsilon(\texttt{A})$ for each module $\texttt{A} \in V$ is satisfied.

The solution to a DQC problem can be concisely characterised as follows:
\begin{definition}[Distribution] \label{def:distribution_hand_wavy}
    Let $Q$ be a set of qubits and $V$ a set of modules. A \textit{distribution} of a quantum circuit on $\lvert Q \rvert$ qubits over a network of $\lvert V \rvert$ modules is characterised by:
    \begin{itemize}
        \item a \textit{qubit allocation} map $\phi \colon Q \to V$ such that $\lvert \phi^ {-1}(\texttt{A}) \rvert \leq \omega(\texttt{A})$ for all $\texttt{A} \in V$ and
        \item an equivalent circuit on $|Q| + \sum_{\texttt{A} \in V} \epsilon(\texttt{A})$ qubits that satisfies the qubit allocation map $\phi$ for the qubits in $Q$ and where all multi-qubit gates between modules are realised via the generation and consumption of ebits.
    \end{itemize}
    We are interested in distributions that consume the fewest number of ebits.
\end{definition}

\section{Background}
\label{sec:background}

\subsection{EJPP protocol and distributable packets}
\label{sec:EJPP}

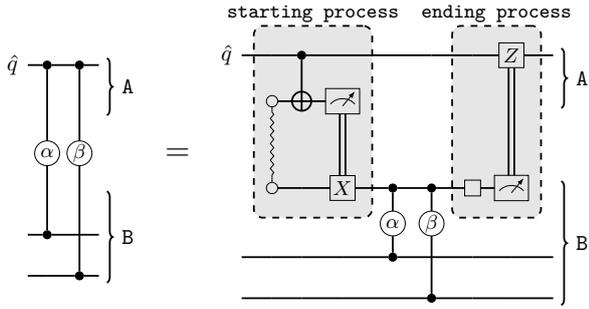
\begin{figure}
    \centering
    \begin{tikzpicture}
  \node[scale=0.85] (eq) {\Large $=$};
  \node[below left=-17mm and 1mm of eq, scale=0.85] (CZs) {
    \begin{quantikz}[column sep=1mm]
        \lstick{$\hat{q}$} & \ctrl{2} & \ctrl{2} & \qw \rstick[2]{\texttt{A}}\\[3pt]
         & & & \\
         & \dPhase{\alpha} & \dPhase{\beta} & \\
         & & & \rstick[3]{\texttt{B}}\\
        \qw & \ctrl{-2} & \qw & \qw \\
        \qw & \qw & \ctrl{-3} & \qw 
    \end{quantikz}
  };
  \node[right=1mm of eq, scale=0.85] (EJPP) {
    \begin{quantikz}[column sep=2mm]
       \lstick{$\hat{q}$} &[5pt] \qw\gategroup[3,steps=3,style={dashed, rounded corners, fill=black!10, inner xsep=2pt}, background]{\small \texttt{starting process}} & \ctrl{1} & \qw &[3pt] \qw & \qw &[3pt] \qw\gategroup[3,steps=2,style={dashed, rounded corners, fill=black!10, inner xsep=2pt}, background]{\small \texttt{ending process}} & \dGate{Z} &[5pt] \qw\rstick[2]{\texttt{A}} \\[-8pt]
       & \dEbit{} \arrow[d, squiggly, no head] & \targ{} & \meter[style={scale=0.75, fill=black!10, thin}]{} \arrow[d, dash, thick, xshift=1.25pt] \arrow[d, dash, thick, xshift=-1.25pt] & & & & & \\[10pt]
       & \dEbit{} & \qw & \dGate{X} & \ctrl{1} & \ctrl{1} & \dH{} & \meter[style={scale=0.75, fill=black!10, thin}]{} \arrow[uu, dash, thick, xshift=1.25pt] \arrow[uu, dash, thick, xshift=-1.25pt] & \rstick[4]{\texttt{B}} \\[-10pt]
       & & & & \dPhase{\alpha} & \dPhase{\beta} & & & \\[-7pt]
      \qw & \qw & \qw & \qw & \ctrl{-1} & \qw & \qw & \qw & \qw \\
      \qw & \qw & \qw & \qw & \qw & \ctrl{-2} & \qw & \qw & \qw
    \end{quantikz}
  };
\end{tikzpicture}
    \caption{\textbf{Distribution of two non-local \CU{} gates via an EJPP protocol.} The gates act on two different modules \texttt{A} and \texttt{B}, the former containing a qubit $\hat{q}$ that both \CU{} gates are applied to. The starting process generates and consumes an ebit, depicted by a wavy line. Other than the ebit, all operations on the distributed circuit are LOCC.}
    \label{fig:ejpp}
\end{figure}

A non-local \CU{} gate can be implemented by consuming a single ebit. The distribution protocol we use originates from~\cite{Eisert_2000} and we refer to it as the EJPP protocol, using the initials of the latter paper's authors. \cref{fig:ejpp} provides an example of such a protocol. During the EJPP protocol, a qubit $\hat{q}$ is shared with a remote module \texttt{B}, entangling it with an ancilla qubit stored in the link qubit register of the module. The starting process of the EJPP protocol --- boxed in grey in~\cref{fig:ejpp} -- generates and consumes an ebit to produce a link qubit that is entangled with $\hat{q}$. The \textit{ending process} only uses LOCC and disentangles the link qubit. Crucially, multiple non-local gates can be implemented using a single EJPP protocol and, hence, consuming a single ebit.

\begin{definition}[Distributable packet] 
    \label{def:packet}
    A \textit{distributable packet} rooted on qubit $\hat{q}$ is a subset of a circuit's non-local \CU{} gates that act on $\hat{q}$ and can all be implemented simultaneously using a single EJPP protocol.\footnote{This definition captures the essence of Definition 16 from~\cite{Junyi2022}. Here, we refer to the elements of a distributable packet $P$ as gates $g \in P$, whereas in~\cite{Junyi2022} the elements of $P$ are pairs $(\hat{q}, t_g)$, where $\hat{q}$ is the qubit that $P$ is rooted on and $t_g$ is the layer in the circuit that gate $g$ appears at. There is an immediate one-to-one correspondence between these two notations; we chose $g \in P$ for the sake of brevity.}
\end{definition}

\begin{lemma} \label{lem:distributable_cond}
Let $P$ be a subset of \CU{} gates in a circuit comprised of gates in \gateset{} for which a qubit allocation map $\phi$ has been provided. If the following three conditions hold, then $P$ is a distributable packet rooted on qubit $\hat{q}$.
    \begin{enumerate}[(a)]
        \item Each gate $g \in P$ acts on $\hat{q}$.
        \item For each $g \in P$ let $q_g$ be the qubit $g$ acts on such that $q_g \not = \hat{q}$; there is a module $\texttt{B} \in V$ such that $\phi(\hat{q}) \not = \texttt{B}$ and $\phi(q_g) = \texttt{B}$ for all $g \in P$.
        \item For every pair of gates $g, g' \in P$, there is no \H{} gate in the circuit acting on $\hat{q}$ between $g$ and $g'$.
    \end{enumerate}
\end{lemma} \begin{proof}
    Conditions (a) and (b) ensure that sharing the state of qubit $\hat{q}$ with module $\texttt{B}$ is sufficient to implement all of the gates in $P$ locally within $\texttt{B}$. A starting process creates a link qubit in module $\texttt{B}$ that is entangled with $\hat{q}$. Then, each gate $g \in P$ is replaced by the same gate acting on $q_g$ and said link qubit. The ending process is applied after the last gate in $P$, measuring out the link and correcting as necessary to guarantee determinism. Condition (c) along with the circuit's gateset \gateset{} implies that all gates between $g$ and $g'$ commute past them. Closer inspection of the circuit for a starting process and ending process (see \cref{fig:ejpp}) reveals these also commute with \RZ{} and \CU{} gates. Therefore, their presence within the EJPP protocol does not affect its operation and we can apply all gates between $g$ and $g'$ unchanged, on their original qubits. Then, it only remains to check the equivalence of the circuits in \cref{fig:ejpp} generalises to the case of any number of consecutive \CU{} gates, which is straightforward.
\end{proof}

\begin{remark} \label{rmk:distributable_cond}
    While conditions (a) and (b) are necessary for all gates in $P$ to be implementable using a single EJPP protocol, (c) can be replaced with a more general condition using a technique known as \textit{embedding}~\cite{Junyi2022}.
    \begin{enumerate}[(c$^*$)]
        \item For every pair of gates $g, g' \in P$, all gates in the circuit acting on $\hat{q}$ between $g$ and $g'$ are \textit{embeddable}.
    \end{enumerate}
    This leads to larger distributable packets. We will discuss what the term \textit{embeddable} refers to in \cref{sec:embedding}. For now, it suffices to know that condition (c) implies (c$^*$).
\end{remark}

\subsection{DQC via hypergraph partitioning}
\label{sec:hypergraph}

The qubit allocation subproblem introduced in \cref{sec:DQC_problem} is reminiscent of a graph partitioning problem. Indeed, we may define the connectivity graph of a circuit as follows: each qubit in the circuit corresponds to a vertex and each \CU{} gate creates an edge between the vertices of the pair of qubits it acts on. It is straightforward to see that partitioning such a graph into $k$ blocks corresponds to allocating each of the qubits to one of $k$ different modules, and cut edges correspond to non-local gates. Thus, the standard graph partitioning problem --- whose goal is to minimise the number of edges cut --- would produce qubit allocations that minimise the number of non-local gates.

However, such an approach would not consider the fact that a single EJPP protocol is capable of implementing multiple non-local gates consuming a single ebit. If our objective is to minimise the number of ebits consumed, a different partition that creates more non-local gates may be advantageous. An example of such a situation is shown in \cref{fig:justifying_hypergraphs}. Crucially, the optimal qubit allocation of \cref{fig:justifying_hypergraphs}a places qubits $q_0$ and $q_1$ in module \texttt{A} and qubits $q_2$ and $q_3$ in module \texttt{B}, but such an assignment differs from the optimal partition of the circuit's connectivity graph \cref{fig:justifying_hypergraphs}b, which would instead place $q_0$ and $q_2$ in module \texttt{A} and qubits $q_1$ and $q_3$ in module \texttt{B}. The former allocation yields four non-local gates whereas the latter yields only three, however, the former can be distributed using two ebits, while the latter requires three.

\begin{figure*}
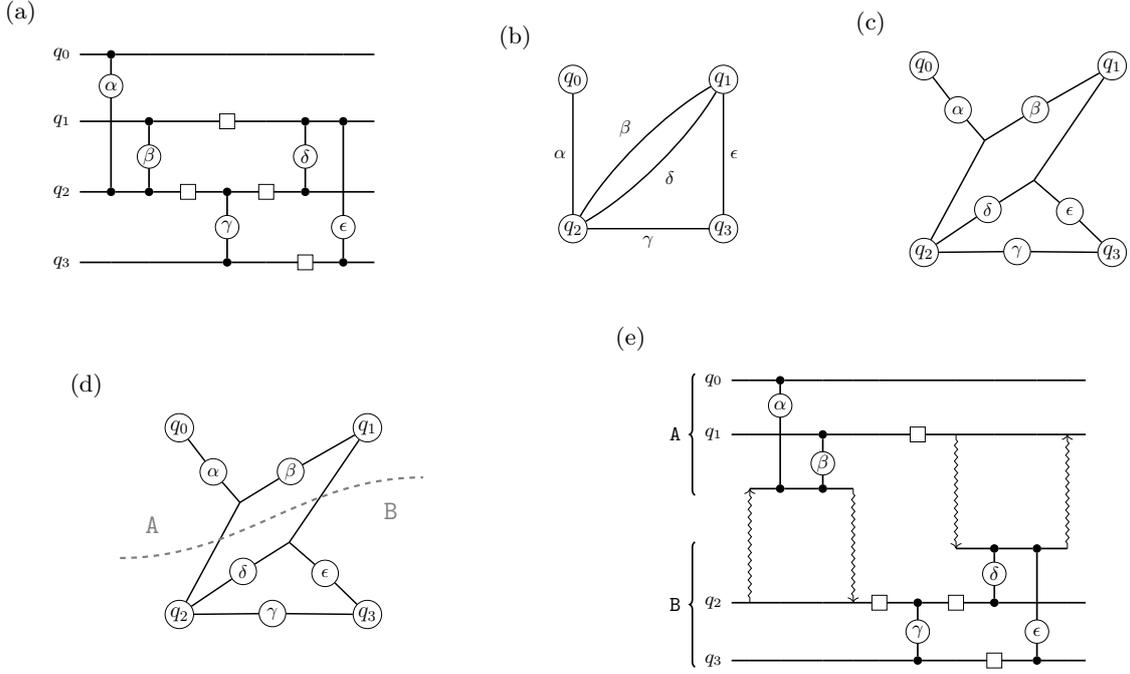

    \centering
    \begin{tikzpicture}
    \node[scale=0.8] (input) {\input{images/justifying_hypergraphs/input_circuit}};
    \node[right=20mm of input, scale=0.7] (graph) {\input{images/justifying_hypergraphs/graph}};
    \node[right=20mm of graph, scale=0.7] (hyp) {\input{images/justifying_hypergraphs/hyp}};
    \node[scale=0.8, below right=10mm and 35mm of input] (dist) {\input{images/justifying_hypergraphs/dist}};
    \node[right=-130mm of dist, scale=0.7] (partition) {\input{images/justifying_hypergraphs/partition}};
    \node[above left=0mm and -1mm of input] {\small (a)};
    \node[above left=0mm of graph] {\small (b)};
    \node[above left=0mm and 1mm of hyp] {\small (c)};
    \node[above left=0mm of partition] {\small (d)};
    \node[above left=0mm of dist] {\small (e)};
\end{tikzpicture}
    \caption{\textbf{Example of correspondence between circuits, graphs and hypergraphs.} (a) An input circuit, (b) its connectivity graph, (c) its hypergraph, as defined in~\cite{Andr_s_Mart_nez_2019}, (d) an optimal partition of the hypergraph, only two hyperedges are cut, (e) the distributed circuit that arises from the hypergraph partition, the number of EJPP protocols matches the connectivity metric of the hypergraph partition: two cuts (Theorem~\ref{thm:reduction}). Starting processes and ending processes are depicted as wavy arrows.}
    \label{fig:justifying_hypergraphs}
\end{figure*}

In~\cite{Andr_s_Mart_nez_2019} it was shown that qubit allocation and non-local gate distribution could both be solved simultaneously via a reduction to hypergraph partitioning. Formally, the only difference between a hypergraph and a graph is that its edges need not be pairs, but subsets of vertices known as `hyperedges'. The intuition behind why hypergraphs are better suited to describe the DQC problem is that, when multiple gates belong to the same distributable packet (Definition~\ref{def:packet}), we may represent them as a single hyperedge. Then, if any number of these gates become non-local due to a qubit allocation, the corresponding effect is that a single hyperedge will be cut by the partition, thus precisely capturing the number of EJPP protocols required to implement the distributable packet. The algorithm that builds such a hypergraph from a given circuit is described in~\cite{Andr_s_Mart_nez_2019} and \cref{fig:justifying_hypergraphs}c shows the outcome of the process on a simple circuit. In~\cite{Andr_s_Mart_nez_2019} the authors proved the following theorem.

\begin{theorem}[\cite{Andr_s_Mart_nez_2019}]
\label{thm:reduction}
    Given a circuit, each of its possible distributed implementations corresponds to a unique partition of its hypergraph. Assuming a fully connected network of modules, the number of ebits required to implement such a distribution coincides with the cost of the partition, calculated using the connectivity metric.\footnote{For a given hypergraph, where $H$ is its set of hyperedges, the connectivity metric~\cite{KaHyPar} of a given partition is calculated as $\sum_{h \in H} \lambda(h)\!-\!1$ where $\lambda(h)$ corresponds to the number of different partition blocks the hyperedge $h$ has vertices on.}
\end{theorem}

This implies that we may reduce the problem of distributing a quantum circuit to the problem of hypergraph partitioning as follows.
\begin{enumerate}
    \item Build the hypergraph of the circuit as described in~\cite{Andr_s_Mart_nez_2019}.
    \item Use a state-of-art hypergraph partitioner to obtain an efficient partition.
    \item Translate the partition into a distribution of the circuit.
\end{enumerate}

Notice that in the hypergraph of \cref{fig:justifying_hypergraphs}c there are more vertices than qubits in the circuit. In fact, there is a vertex per qubit and a vertex per \CU{} gate; we call them \textit{qubit-vertices} and \textit{gate-vertices} respectively. When a partition assigns a qubit-vertex to block $\texttt{A}$ it indicates that such a qubit is to be allocated to module $\texttt{A}$; similarly, when a gate-vertex is assigned to block $\texttt{A}$ it indicates that the corresponding \CU{} gate ought to be implemented as a local \CU{} gate within module $\texttt{A}$, with the aid of an EJPP protocol if the hyperedge is cut. \cref{fig:justifying_hypergraphs}d and \cref{fig:justifying_hypergraphs}e exemplify how a partition of the hypergraph gives rise to a distribution of the original circuit.

This approach, as presented in~\cite{Andr_s_Mart_nez_2019} has some shortcomings, identified below.
\begin{itemize}
    \item Hypergraph partitioners often assume that all blocks of the partition should be filled with approximately the same number of vertices each. In our task, however, each module $\texttt{A}$ may have a different capacity of workspace qubits $\omega(\texttt{A})$. To account for this constraint, we may use hypergraph partitioners such as \KaHyPar{}~\cite{KaHyPar} which allow us to indicate the maximum capacity of each block; more details on \cref{sec:hypergraph_heterogeneous}.
    \item In~\cite{Andr_s_Mart_nez_2019}, only fully connected networks were considered (\ie{} complete graphs), whereas in this work we consider heterogeneous networks (\cref{sec:DQC_problem}). Creation and sharing of ebits between adjacent modules is directly supported by the network's hardware. Non-adjacent modules may still share an ebit, but producing it will require some entanglement distribution, consuming multiple hardware-supported ebits in the process. The framework summarised in this section is not capable of making such a distinction. Some techniques used to extend hypergraph partitioning to account for heterogeneous networks are discussed in \cref{sec:hypergraph_heterogeneous}.
    \item In \cref{sec:embedding} we discuss some advanced techniques to further reduce the ebit count of implementing non-local gates by merging multiple distributable packets. These techniques are beyond what can be captured in terms of hypergraph partitioning. Thus, in \cref{sec:solutions} we use the hypergraph partitioning to provide an initial solution to the DQC problem, whose non-local gate distribution is later refined using the techniques from \cref{sec:embedding}, \cref{sec:vertex_cover_embedding} and \cref{sec:refinement}.
\end{itemize}

Approaches that solve the two subproblems of DQC separately --- qubit allocation and non-local gate distribution --- have been proposed in the literature. In~\cite{gsundaram_et_al}, the authors solve the qubit allocation subproblem by partitioning a weighted graph describing the connectivity of the circuit, where the calculation of the weights attempts to take into account cases where multiple non-local gates may be implemented using a single ebit. Such an approach has the same shortcomings listed above, with the additional drawback that the weights only provide an estimate for the ebit cost (rather than the exact value as in the case of hypergraph partitioning) and the advantage that graph partitioners are simpler and, hence, can be expected to perform better than hypergraph partitioners. A follow up paper by the same authors \cite{Sundaram2022GeneralDistribution} solves the qubit allocation subproblem using a Tabu search algorithm. The latter work supports heterogeneous networks, solving one of the three shortcomings discussed above, but fails to take advantage of the optimisation opportunities we discuss in \cref{sec:Steiner}. Both of these works solve non-local gate distribution on a second step, which we review in \cref{sec:vertex_cover}.

\subsection{Embedding}
\label{sec:embedding}

Lemma~\ref{lem:distributable_cond} provides sufficient conditions for a group of non-local \CU{} gates to belong to the same distributable packet. Remark~\ref{rmk:distributable_cond} hinted at a more general condition involving the notion of \textit{embedding} proposed in~\cite{Junyi2022}. \cref{fig:embedding_example} provides a couple of examples where the \CU{} gates of phase $\alpha$ and $\beta$ belong to the same distributable packet even though there are \H{} gates between them, violating condition (c) from Lemma~\ref{lem:distributable_cond}.

\begin{figure}
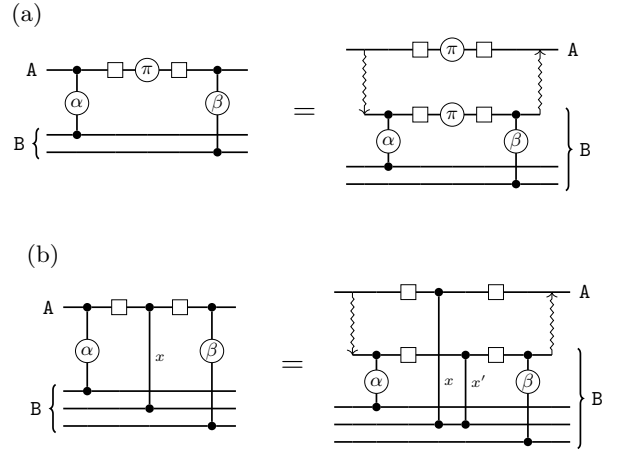

    \centering
    \begin{tikzpicture}
  \node[scale=0.8] (eq1) {\Large $=$};
  \node[below left=30mm and -3.5mm of eq1, scale=0.8] (eq2) {\Large $=$};
  \node[below left=-11mm and 3mm of eq1, scale=0.8] (HZH) {\input{images/embedding_example/HZH}};
  \node[below right=-13mm and 1mm of eq1, scale=0.8] (HZH_emb) {\input{images/embedding_example/HZH_emb}};
  \node[below left=-12.5mm and 3mm of eq2, scale=0.8] (HCZH) {\input{images/embedding_example/HCZH}};
  \node[below right=-14.5mm and 1mm of eq2, scale=0.8] (HCZH_emb) {\input{images/embedding_example/HCZH_emb}};
  \node[above left=1mm and -7mm of HZH] {\small (a)};
  \node[above left=1mm and -7mm of HCZH] {\small (b)};
\end{tikzpicture}
    \caption{\textbf{Examples of embedding.} (a) Embedding $\H \cdot \Z \cdot \H$, (b) embedding $\H \cdot \CZ \cdot \H$, the correction \CZ{} gate $x'$ is local. In both examples, the \CU{} gates with phase $\alpha$ and $\beta$ are implemented by the same EJPP protocol.}
    \label{fig:embedding_example}
\end{figure}

\begin{definition}[Embedding unit] \label{def:embedding_unit}
    Consider an EJPP protocol with starting process $\mathcal{S}_{\hat{q},\texttt{B}}$ sharing qubit $\hat{q}$ with module $\texttt{B}$ and ending process $\mathcal{E}_{\hat{q},\texttt{B}}$ (see \cref{fig:ejpp}). An embedding unit is a subcircuit $C$ satisfying the following identity:
    \begin{equation}
        \mathcal{E}_{\hat{q},\texttt{B}} \, C \, \mathcal{S}_{\hat{q},\texttt{B}} \ = \ \left( \bigotimes_{\texttt{A} \in V} L_\texttt{A} \right) C \left( \bigotimes_{\texttt{A} \in V} K_\texttt{A} \right)
    \end{equation}
    where $V$ is the set of modules in the network and for each module $\texttt{A} \in V$, $L_\texttt{A}$ and $K_\texttt{A}$ are local gates within $\texttt{A}$. We refer to the gates $L_\texttt{A}$ and $K_\texttt{A}$ as the \textit{correction gates} of the embedding.
\end{definition}

In essence, an embedding unit is a subcircuit appearing between gates of a distributable packet $P$ such that, if $P$ is distributed, we only require local correction gates to maintain circuit equivalence. Importantly, notice that we do not require $C$ to be local --- it has not yet been distributed. Indeed, the embedded \CZ{} gate labelled $x'$ in \cref{fig:embedding_example}b is non-local.

It is straightforward from the above definition that any gate that commutes with a starting process $\mathcal{S}_{\hat{q},i}$ forms an embedding unit by itself, which is the reason why condition (c) from Lemma~\ref{lem:distributable_cond} implies (c$^*$) from Remark~\ref{rmk:distributable_cond}. More interesting embedding units containing \H{} gates are captured by the following lemma.

\begin{lemma} \label{lem:embedding_cond}
    Let $C$ be a circuit built from \gateset{} containing a qubit $\hat{q}$, let $\texttt{B}$ be a module and let $\phi$ be a qubit allocation such that $\phi(\hat{q}) \not = \texttt{B}$. If each of the following conditions holds, then $C$ is an embedding unit of an EJPP protocol sharing $\hat{q}$ with module $\texttt{B}$.
    \begin{enumerate}[(a)]
        \item The first gate and last gates in $C$ are \H{} gates acting on $\hat{q}$.
        \item All \CU{} gates within $C$ that act on $\hat{q}$ have their other qubit allocated to module $\texttt{B}$.
        \item All \CU{} gates within $C$ that act on $\hat{q}$ have $\pi$ phase --- \ie{} they are \CZ{} gates. 
        \item All \RZ{} gates within $C$ that act on $\hat{q}$ may be squashed together so that only \RZ{} gates with $\pi$ phase remain --- \ie{} Pauli \Z{} gates.
        \item There are no more than two \H{} gates acting on $\hat{q}$ in $C$.
    \end{enumerate}
\end{lemma} \begin{proof}
    Immediate from Corollary 30 of~\cite{Junyi2022}. Alternatively, this is a straightforward generalisation of the two embedding units shown in \cref{fig:embedding_example}.
\end{proof}

Lemma~\ref{lem:embedding_cond} provides a sufficient condition for a subcircuit to be an embedding unit. In \cite{Junyi2022} a more detailed analysis shows that condition (e) can be relaxed, but the formalisation of this more general condition is too intricate to be presented in this summary. These more general conditions can be checked on a circuit using Algorithm 35 from~\cite{Junyi2022}, which we implemented in the software \pytketdqc{} we present in this work.

Equipped with the notion of embedding units and the algorithm from~\cite{Junyi2022} to identify them, we can now build larger distributable packets. Whenever a gate commutes with the packet's starting process, embedding it requires no correction gates. Whenever we encounter an embedding unit, we apply the embedding rules from Corollary 30 of~\cite{Junyi2022} to introduce the required local correction gates; more details are provided in \cref{sec:to_pytket_circuit}.

Condition (b) from Lemma~\ref{lem:embedding_cond} has a rather subtle implication: if two embedding units on different qubits contain the same \CZ{} gate, only one of the two is embeddable, see \cref{fig:embedding_conflict}. We then say that such a pair of embedding units have an \textit{embedding conflict}; similarly, two distributable packets that contain embedding units in conflict are also said to have an embedding conflict. Resolving an embedding conflict consists of choosing which of the two distributable packets should be distributed and splitting the other one into two separate packets so that embedding the conflicting \CZ{} gate a second time is no longer necessary. An algorithm for non-local gate distribution that takes advantage of embedding and resolves embedding conflicts was proposed in \cite{Junyi2022}; we briefly review the algorithm in \cref{sec:vertex_cover_embedding}.

\begin{figure*}
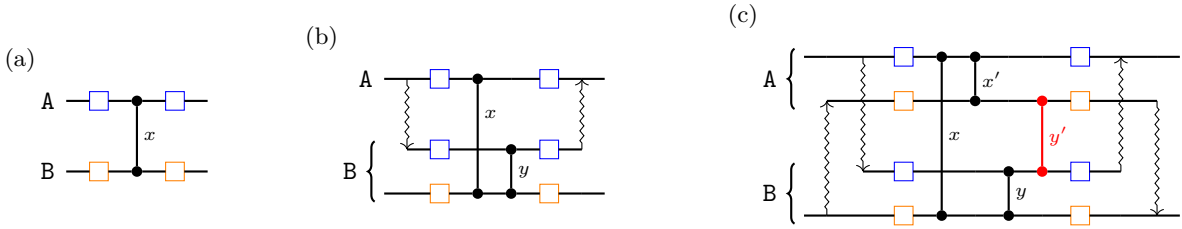

    \centering
    \begin{tikzpicture}
    \node (a) {\input{images/embedding_conflict/a}};
    \node[right=13mm of a] (b) {\input{images/embedding_conflict/b}};
    \node[right=13mm of b] (c) {\input{images/embedding_conflict/c}};
    \node[above left=-1mm and -3mm of a] {\small (a)};
    \node[above left=-1mm and -3mm of b] {\small (b)};
    \node[above left=-1mm and -3mm of c] {\small (c)};
\end{tikzpicture}
    \caption{\textbf{Embedding conflict.} (a) A simple circuit with two embedding units: one containing the \CZ{} gate and the blue \H{} gates; the other containing the \CZ{} gate and the orange \H{} gates. (b) Embedding only one of the embedding units causes no issues: the correction gate $y$ is local. However, notice that the other embedding unit (the one containing the orange \H{} gates) contains $y$ as well. (c) Since $y$ does not satisfy Lemma~\ref{lem:embedding_cond} (both of its qubits are in the same module), embedding it would create a non-local correction gate $y'$, defeating the purpose of embedding.}
    \label{fig:embedding_conflict}
\end{figure*}

For the sake of brevity, we have not discussed how to deal with situations where a certain embedding unit must be embedded within more than one distributed packet. Thanks to Corollary 14 from~\cite{Junyi2022}, we know that these situations will never cause new conflicts. A more subtle situation arises when two distributable packets $P$ and $P'$ happen to be intertwined in the sense that some gate $g \in P$ needs to be embedded within $P'$ while, at the same time, some other gate $g' \in P'$ needs to be embedded within $P$. We describe how we deal with such a situation in \cref{sec:to_pytket_circuit}.

\subsection{Non-local gate distribution via vertex cover}
\label{sec:vertex_cover}

In this section we review the literature on the subproblem of non-local gate distribution. We focus on approaches that reduce it to finding the minimum vertex cover of a graph. We begin from a version of the problem with the following simplifications:
\begin{itemize}
    \item we assume that the network of modules is fully connected,
    \item we ignore the optimisation opportunities embedding provides and
    \item we impose that a non-local gate must be implemented in either of the two modules it acts on --- unfortunately, this prevents beneficial distributions such as the one in \cref{fig:detached_gate} from being considered.
\end{itemize}
This simplified problem is presented in~\cite{gsundaram_et_al} under the name \texttt{MS-HC}; we summarise their solution in this section.
One of the contributions of the present work is the extension of their approach to the general problem where these three constraints are lifted. In particular, \cref{sec:ALAP} and \cref{sec:refinement steiner} allows us to consider heterogeneous networks of modules, we use the approach of~\cite{Junyi2022} (summarised in \cref{sec:vertex_cover_embedding}) to exploit embedding and employ the method in \cref{sec:refinement detatched} to lift the last of the constraints.

\begin{figure}
    \centering
    \begin{tikzpicture}
    \node[scale=0.9] (eq) {\Large $=$};
    \node[left=3mm of eq, scale=0.9] {
        \begin{quantikz}[column sep=3mm]
            \lstick[1]{\texttt{A}} & \ctrl{1} & \ctrl{1} & \qw & \qw \\[-7pt]
            & \dPhase{\alpha} & \dPhase{\beta} & & \\[-7pt]
            \lstick[1]{\texttt{B}} & \ctrl{-1} & \qw & \ctrl{1} & \qw \\[-7pt]
            & & & \dPhase{\gamma} & \\[-7pt]
            \lstick[1]{\texttt{C}} & \qw & \ctrl{-3} & \ctrl{-1} & \qw 
        \end{quantikz}
    };
    \node[right=2mm of eq, scale=0.9] {
        \begin{quantikz}[column sep=3mm]
            \lstick[1]{\texttt{A}} & \qw \arrow[d, squiggly] & \qw & \qw & \qw & \qw & \qw \\[5pt]
            \lstick[5]{\texttt{B}} & & \ctrl{1} & \ctrl{1} & \qw & \qw \arrow[u, squiggly] & \\[-10pt]
            & & \dPhase{\alpha} & \dPhase{\beta} & & & \\[-10pt]
            \qw & \qw & \ctrl{-1} & \qw & \ctrl{1} & \qw & \qw \\[-10pt]
            & & & & \dPhase{\gamma} & & \\[-10pt]
            & & \qw & \ctrl{-3} & \ctrl{-1} & \qw \arrow[d, squiggly] & \\[5pt]
            \lstick[1]{\texttt{C}} & \qw \arrow[u, squiggly] & \qw & \qw & \qw & \qw & \qw 
        \end{quantikz}    
    };
\end{tikzpicture}
    \caption{\textbf{Distribution with detached gate.} In the distributed circuit, the \CU{} gate with phase $\beta$ is implemented within module \texttt{B}, but originally had none of its qubits allocated to it --- we refer to it as a detached gate.}
    \label{fig:detached_gate}
\end{figure}
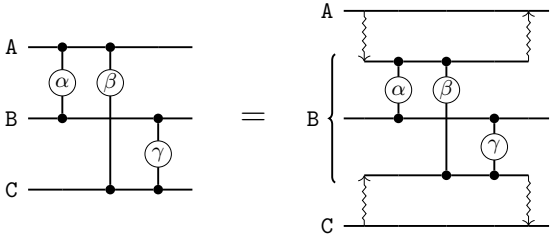

Once a qubit allocation has been chosen, all that remains to do is identify distributable packets and, for each non-local \CU{} gate, decide which of the packets it belongs to should be used to distribute it. In the absence of embedding, the first task --- which we refer to as \textit{gate packing} --- is trivial: scan the circuit qubit by qubit, from beginning to end, and find sequences of gates satisfying Lemma~\ref{lem:distributable_cond}. We shall only consider the collection of largest distributable packets, \ie{} those that are not a subset of any other distributable packet; as a consequence, each non-local \CU{} belongs to exactly two distributable packets --- one per qubit.

The second task corresponds to finding the minimum vertex cover of a graph whose vertices represent the distributable packets and where an edge between two of these corresponds to the existence of at least one non-local \CU{} gate contained in both packets. A vertex cover of a graph is a subset of its vertices such that each edge is incident to at least one vertex in the subset; thus, a vertex cover of the previous graph selects which distributable packets ought to be realised so that all non-local \CU{} gates are distributed. A minimum vertex cover of the graph would select the fewest number of distributable packets and, hence, yield the optimal distribution under the given constraints. In the appendix of~\cite{gsundaram_et_al} its authors show that the previous graph is guaranteed to be bipartite, which implies the minimum vertex cover can be found efficiently. 

The authors of~\cite{gsundaram_et_al} then considered a more general problem where non-local gates may be implemented in a detached manner, \ie{} so that distributions such as the one from \cref{fig:detached_gate} may be explored. This more general problem is not known to be reducible to a vertex cover problem on a bipartite graph. Nevertheless, the authors of~\cite{gsundaram_et_al} proposed an efficient algorithm that is guaranteed to provide a distribution only a logarithmic factor away from the best distribution achievable under the given constraints. However, said algorithm is still solving a simplified problem since it is omitting the following constraints and optimisation opportunities.
\begin{description}
    \item[Network topology:] In \cref{sec:DQC_problem} we let networks be described by arbitrary (connected) graphs. Thus, the approach should take into account the distance between modules when computing the cost of distributing non-local gates.
    \item[Bounded link qubit register:] Each module $\texttt{A}$ may have a bound to the size of its link qubit register $\epsilon(\texttt{A})$ (see \cref{sec:DQC_problem}). The resulting distribution should refrain from exceeding it.
    \item[Embedding:] The embedding technique described in \cref{sec:embedding} lets us create larger distributable packets. Thus, an algorithm using embedding is likely to cover all non-local gates using fewer packets and, hence, find a distribution that uses fewer ebits.
\end{description}

In \cite{Sundaram2022GeneralDistribution} the authors propose an algorithm that is aware of the network topology and the bound to the size of the link qubit register. Rather than a minimum vertex cover problem, they consider the dual problem of maximising the number of non-local gates covered using a fixed number of ebits while satisfying a set of linear constraints. The linear constraints are used to capture the network topology and the bound to the link qubit register. The central element of the optimisation procedure is carried out by an integer linear programming (ILP) subroutine. 
However, this approach does not take advantage of embedding.

In \cite{Junyi2022} an algorithm that exploits embedding is proposed. The algorithm is based on finding minimum vertex covers and it makes use of graph colouring to identify solutions that satisfy the bound to the link qubit register. However, the algorithm is targeted to networks containing only two modules and, consequently, has a trivial network topology. The approach to DQC we propose in the present paper takes multiple insights from this latter work, orchestrating them in a more general framework. Thus, we dedicate the following section to introduce the ideas from \cite{Junyi2022} that are relevant to us.

\subsubsection{Embedding-aware approach}
\label{sec:vertex_cover_embedding}

The approach to non-local gate distribution using minimum vertex cover can be extended to account for embedding. The means to do so were described in \cite{Junyi2022}, where detailed algorithms were provided. Once again, the first step is to identify the largest distributable packets that can be realised without the use of embedding. Then, for each distributable packet $P$, we check whether the gates that come immediately after $P$ form an embedding unit. If so, this would allow to merge $P$ with the packet appearing immediately after the embedding unit, creating a larger distributable packet. The algorithm identifies the largest distributable packets that can be achieved by such merging, and records the embeddings that are required to do so. This task simply requires us to carry out a scan over the circuit and, hence, it scales linearly with the dimensions of the circuit. 

The resulting distributable packets are then arranged in a graph $G$ as in the case of the standard vertex cover approach: its vertices correspond to each of the packets and its edges correspond to common non-local \CU{} gates between them. It may seem that it only remains to find a minimum vertex cover of $G$, but this would not account for embedding conflicts (see \cref{fig:embedding_conflict}). Instead, we need to define an additional graph $K$ whose vertices correspond to the embeddings that were used when merging distributable packets, where an edge between two such embeddings appears if and only if the embeddings are in conflict. With these two graphs $G$ and $K$ at hand, a sketch of the algorithm is presented below.

\begin{enumerate}
    \item Find a minimum vertex cover $\mathcal{C}_G$ of $G$.
    \item Find the subset $\kappa$ of embeddings required to implement all of the distributable packets in $\mathcal{C}_G$.
    \item Extract the subgraph $K_\kappa$ of $K$ whose vertex set is $\kappa$ and whose edges are those from $K$ that connect vertices in $\kappa$. 
    \item Obtain a minimum vertex cover $\mathcal{C}_K$ of $K_\kappa$: this is the smallest set of embeddings that we must give up in order to resolve all embedding conflicts incurred by $\mathcal{C}_G$.
    \item For each element in $\mathcal{C}_K$ --- an embedding ---, identify which distributable packet $P \in \mathcal{C}_G$ used it (there is exactly one) and update $\mathcal{C}_G$ replacing $P$ with two distributable packets: one containing all of the gates in $P$ that come before the \CZ{} gate responsible for the embedding conflict and another with the gates in $P$ that come afterwards.
\end{enumerate}

The resulting set of distributable packets $\mathcal{C}_G$ is no longer a \textit{minimum} vertex cover of $G$, but it is a vertex cover with no embedding conflicts. Thus, it can be used to generate a valid distribution. This approach is not guaranteed to return the overall optimal solution, but it does resolve the embedding conflicts of a given vertex cover of $G$ in an optimal way. In an attempt to find better overall solutions, we may choose to repeat the routine above for multiple distinct vertex covers of $G$ and pick the best among them \cite{Junyi2022}.

The algorithms presented in \cite{Junyi2022} for the tasks just described were designed for networks with exactly two modules. Generalising these to networks of multiple modules is immediate: it is sufficient that our conditions for identifying distributable packets (Lemma~\ref{lem:distributable_cond}) and embedding units (Lemma~\ref{lem:embedding_cond}) required that all of their \CU{} gates acted on the same two modules. This guarantees that both $G$ and $K$ are bipartite graphs, so we may find a minimum vertex cover for them efficiently. The fact that these graphs are bipartite graphs is not trivial, but it follows from the same argument the authors of \cite{gsundaram_et_al} used for their bipartite graph for the \texttt{MS-HC} problem. 

The authors of \cite{Junyi2022} propose how to take into account the bound to the link qubit register size --- via graph colouring -- solutions that exceed such a bound. Then they present an efficient way of splitting the offending distributable packets so that the number of EJPP protocols that are simultaneously active is reduced, at the cost of increasing the total number of ebits consumed. Such an approach is beyond the scope of the present paper and we omit further details for the sake of brevity.

\subsection{Intermediate representation of distribution}
\label{sec:distribution_as_hypergraphs}

Throughout this section we have discussed multiple approaches aimed at optimising different aspects of the DQC problem. We have considered multiple abstractions --- \eg{} distributable packets, embedding units, embedding conflicts, hypergraphs, \textit{etc\@.} --- each tailored to be as natural as possible to the approach at hand. Our goal in this paper is to propose an approach that can take advantage of the insights of each of these optimisation methods. To do so, we require an intermediate representation where the outcome of each of these optimisation methods can be represented. Such an intermediate representation could simply be a partially distributed circuit; however, a more abstract representation would be preferable to minimise the overhead of dealing with superfluous low level details --- such as the correction gates required for an embedding unit, the exact placement of a starting process within a circuit, the reuse of link qubits, \textit{etc\@.} --- that could easily be deferred to the final step of the workflow. Fortunately, all of what has been reviewed in this section can be captured within the framework of hypergraphs discussed in \cref{sec:hypergraph}, which makes it a natural choice for our intermediate representation of a distribution.

\begin{definition}[IR of distributions] \label{def:distribution}
    A \texttt{Distribution} contains the following information.
    \begin{itemize}
        \item A hypergraph of $\lvert Q \rvert \!+\! \lvert \mathcal{G} \rvert$ vertices, where $Q$ is the set of qubits in the original circuit and $\mathcal{G}$ is its collection of \CU{} gates. We refer to these as qubit-vertices and gate-vertices respectively, as established in \cref{sec:hypergraph}.
        \item An allocation map $\phi \colon Q \cup \mathcal{G} \to V$, where $V$ is the set of modules in the network.
    \end{itemize}
    Additionally, we include the original circuit and the network of modules, which remain unchanged throughout the workflow.
\end{definition}

The purpose of including the original circuit and the network of modules within the \texttt{Distribution} is to be able to assess the ebit cost (see \cref{sec:ALAP}). Furthermore, the information contained in \texttt{Distribution} is all we require to generate the corresponding distributed circuit; we explain how to do so in Appendix~\ref{sec:to_pytket_circuit}. Notice that the allocation map $\phi$ determines a partition of the hypergraph. Below, we briefly discuss how the different abstractions considered in this section can be captured within a \texttt{Distribution}. Recall that, by construction of the hypergraph in \cref{sec:hypergraph}, each hyperedge has a single qubit-vertex and each gate-vertex is present in exactly two hyperedges.

\begin{description}
    \item[Non-local gate:] a gate $g \in \mathcal{G}$ is non-local if and only if its adjacent qubit-vertices $q$ and $q'$ satisfy $\phi(q) \not = \phi(q')$.
    \item[Detached gate:] a gate $g \in \mathcal{G}$ is detached if and only if its adjacent qubit-vertices $q$ and $q'$ satisfy $\phi(q) \not = \phi(g)$ and $\phi(q') \not = \phi(g)$.
    \item[Distributable packet:] a distributable packet $P$ rooted on $\hat{q}$ can be represented as a hyperedge with qubit-vertex $\hat{q}$ and the gate-vertices corresponding to the gates in $P$. In general, a hyperedge may contain the union of any number of distributable packets as long as they are all rooted on the same qubit $\hat{q}$. Whereas it is necessary for all gates of a distributable packet $g \in P$ to be allocated to the same module $\phi(g)$, this requirement does not apply to hyperedges. As a consequence, we can extract the distributable packets comprising a hyperedge by grouping its gate-vertices in terms of where they are allocated to.
    \item[Embedding unit:] if two distributable packets may be merged together by embedding the gates between them, the same can be said about merging the hyperedges the packets belong to. As such, embedding techniques alter the hypergraph itself, increasing the size of hyperedges for the sake of reducing their number. Embedding units can be retrieved on demand by inspecting the subcircuit between any two gates on the same hyperedge. Since \texttt{Distribution} is meant to capture valid distributions, we assume no embedding conflicts are incurred; it is the responsibility of the optimising method to guarantee that this is satisfied.
\end{description}

Verifying that the bound to computation registers $\omega(\texttt{A})$ of each module $\texttt{A} \in V$ is satisfied is straightforward: simply count the number of $q \in Q$ such that $\phi(q) = \texttt{A}$. The cost in the number of ebits can be inferred using the methods presented in \cref{sec:ALAP}. Unfortunately, the satisfaction of bound to link qubit registers $\epsilon(\texttt{A})$ cannot be easily checked using our intermediate representation; instead, we need to generate its corresponding distributed circuit (as detailed in Appendix~\ref{sec:to_pytket_circuit}) and count the number of link qubits used --- recall that this is not the same as the number of ebits, since space in the link qubit registers may be reused. This is not an obstacle to our optimisation approaches since none of them consider the bound $\epsilon(\texttt{A})$ within their routines: satisfaction of this bound is deferred to a final pass at the end of the workflow that acts directly on the distributed circuit and is described in Appendix~\ref{sec:limited links}.

\section{Distribution techniques}
\label{sec:solutions}

In this section we discuss the novel distribution techniques that we have implemented in \pytketdqc{}, our DQC tool, available at \repo{}. Our tool is designed as an extension to \texttt{pytket}, the Python interface of the TKET compiler \cite{Sivarajah_2020} and, as such, it may easily be integrated in a full compilation stack. 

Our techniques are orchestrated together in the default workflows detailed on \cref{sec:benchmarks workflows}. The user may choose to run these default workflows or create a custom one, combining the distribution techniques available as they prefer. Any DQC workflow making use of \pytketdqc{} should contain the following steps, in this precise order.

\begin{description}
    \item[Rebase.] Rewrite the circuit to an equivalent one in the gateset \gateset{}. Within \pytketdqc{} we provide an automated method to do so, based upon the rebase passes provided within \texttt{pytket}.
    
    \item[Qubit allocation.] Assign to which module each qubit of the circuit should be allocated, adhering to the bound on the size of the computation register. Our techniques are based on the hypergraph representation discussed in \cref{sec:distribution_as_hypergraphs}, and the user may choose between an annealing approach or a third-party hypergraph partitioner with a greedy refinement, both of which are detailed in \cref{sec:hypergraph_heterogeneous}. Both of these take advantage of Steiner trees as discussed in \cref{sec:Steiner}.
    
    \item[Gate packing.] This step is meant to identify opportunities where embedding may be used, passing this information to the next step. In particular, we implemented the algorithm proposed in \cite{Junyi2022} for this task, whose core ideas are summarised in \cref{sec:embedding}.
    
    \item[Non-local gate distribution.] Two options are available: either use the solution provided by the qubit allocation step --- distribute gates according to which modules their gate-vertices are assigned to --- or make use of the vertex cover approach proposed in \cite{Junyi2022} and summarised in \cref{sec:vertex_cover_embedding}. The former option will not take advantage of embedding, but will make use of Steiner trees; conversely, the latter option will consider embedding but not Steiner trees. Neither of these guarantee satisfaction of the bound to the link qubit registers; this is deferred to the last step of the workflow.
    
    \item[Refinement.] The previous step makes use of either the embedding technique or Steiner trees. During this refinement step, the user can choose to apply any number of the passes described in \cref{sec:refinement}. These refinement passes further improve upon the current solution by taking advantage of readily available opportunities for optimisation using Steiner trees and embedding. The key insight that lets us combine these two seemingly mutually exclusive techniques is described in \cref{sec:ALAP}. A refinement that lets us take advantage of detached gates (as in \cref{fig:detached_gate}) is also provided.
    
    \item[Circuit generation.] Our tool provides methods for the automatic generation of the distributed circuit as a \texttt{pytket} circuit or QASM file. We keep track of the occupancy of the link qubit register of each module and reuse link qubits after the EJPP protocol that employed them terminates. Thus, even though our methods do not guarantee satisfaction of a bound to communication memory, the required memory capacity is not directly dependent on the number of EJPP protocols carried out, but rather on the maximum number of EJPP protocols simultaneously active at any given time. As shown in \cref{sec:limited links}, the size of the link qubit registers remains manageable, even if the user does not specify a bound. If the user does specify a bound to link qubit registers, we use the routine described in \cref{sec:limited links} to update the distributed circuit as necessary to satisfy the bound, at the cost of increasing the number of ebits required.
\end{description}

Moreover, our tool provides some basic functions for analysing the distributed circuit, such as counting the number of ebits used and the qubit occupancy of the registers of each module. We also provide a method to verify the equivalence between the original circuit and the distributed one, based on~\cite{PyZX}, which is automatically called at the end of the circuit generation step. 

\subsection{Gate distribution using Steiner trees}
\label{sec:Steiner}

One approach to implementing a distribution hyperedge between two non adjacent modules in a heterogeneous network would be to first construct a single ebit between the relevant modules. This could be done via entanglement swapping; consuming ebits between intermediate modules in the network to build the single required ebit. This single ebit can then be used to perform the EJPP protocol at a total cost in ebits equal to the shortest path in the network between the two modules. In the case where the hyperedge is distributed between three modules, which is to say two distributable packets, and so EJPP processes, are required, the e-bit cost of this approach is the sum of the cost of constructing two ebits. In this case this would be the sum of the shortest paths in the network between the module from which the hyperedge is being distributed, and the two other modules.

During the above described technique, the proxy link qubits in the intermediate modules are measured before the non-local gates have been applied. Alternatively, as these disentangling operation do not affect the qubits which are acted on by the non-local gate, they may be delayed until after the non-local gates have been enacted. Additionally, the starting and ending process commute with the controls of the distributed gates. This means that when non-local gates belong to the same hyperedge are distributed to separate modules, all starting processes can be performed before the gates are enacted, and all ending process may be performed after all gates are acted. This process is depicted in \cref{fig:steiner ejpp}.

Reusing intermediate link qubits in the aforementioned way reduces the e-bit cost of the distribution to the size of the smallest subtree of the module network which includes the modules of concern. This subgraph is known as a Steiner tree. This approach extends to Steiner trees of arbitrary shape, as exemplified in \cref{fig:steiner ejpp}. Circuit distribution in \pytketdqc{} makes use of Steiner trees instead of entanglement swapping, allowing us to make savings upon a naive application of the EJPP protocol.

\begin{figure*}
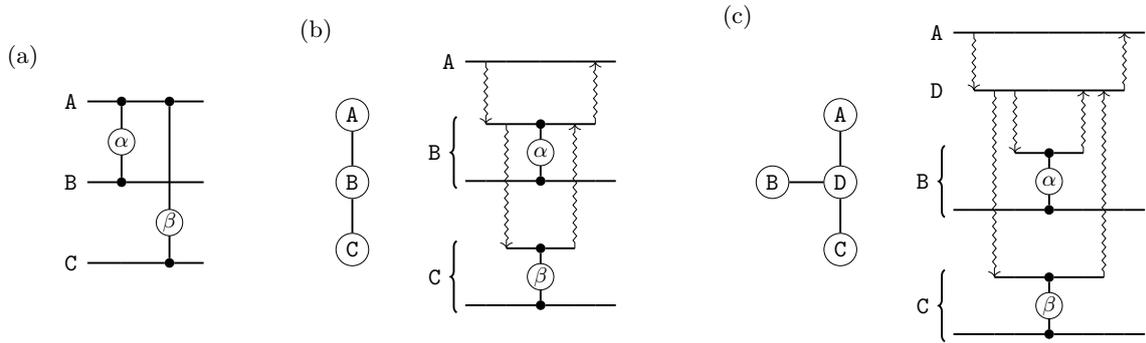

    \centering
    \begin{tikzpicture}
    \node[scale=0.9] (input) {\input{images/steiner/input}};
    \node[right=13mm of input, scale=0.9] (line) {\input{images/steiner/line}};
    \node[right=13mm of line, scale=0.9] (T) {\input{images/steiner/T}};
    \node[above left=0mm and 0mm of input] {\small (a)};
    \node[above left=-3mm and -3mm of line] {\small (b)};
    \node[above left=-5mm and -3mm of T] {\small (c)};
\end{tikzpicture}
    \caption{\textbf{Gate distribution via EJPP embedding with Steiner trees.} (a) An input circuit where each qubit is allocated  to a different module. (b) The distribution of (a) onto a line network. This should be compared to the approach of using entanglement swapping, where the number of required ebits would be 3. In this case disentangling operation have been delayed until after the non-local gates have been enacted, and starting processes related to $\beta$ have been commuted to the start of the computation. This results in a total ebit cost of 2 (c) The distribution of (a) in a T shaped network with entangling and disentangling operations acting along the edges of the Steiner tree connecting the relevant modules.}
    \label{fig:steiner ejpp}
\end{figure*}

Note that it is not possible to safely commute entangling and disentangling operation as described in the case when Steiner trees are combined with embedding units containing \H{} gates. This is discussed in in \cref{sec:ALAP}. 



\subsection{Combining embedding and Steiner trees}
\label{sec:ALAP}

The approach proposed in \cref{sec:Steiner} efficiently generates the entanglement sharing required for the distribution of the gates in a hyperedge, using Steiner trees. To do so, we maintain the entanglement of some proxy link qubits throughout the whole duration of the collection of EJPP processes. Unfortunately, if the hyperedge includes any distributable packet that requires some embedding, such as the example in \cref{fig:ALAP}, maintaining the entanglement of these proxy link causes a problem: correction gates acting on them will be required. As shown in \cref{fig:ALAP} these correction gates may be non-local, thus creating the need for extra ebits to implement them, defeating the purpose of embedding. 

\begin{figure*}
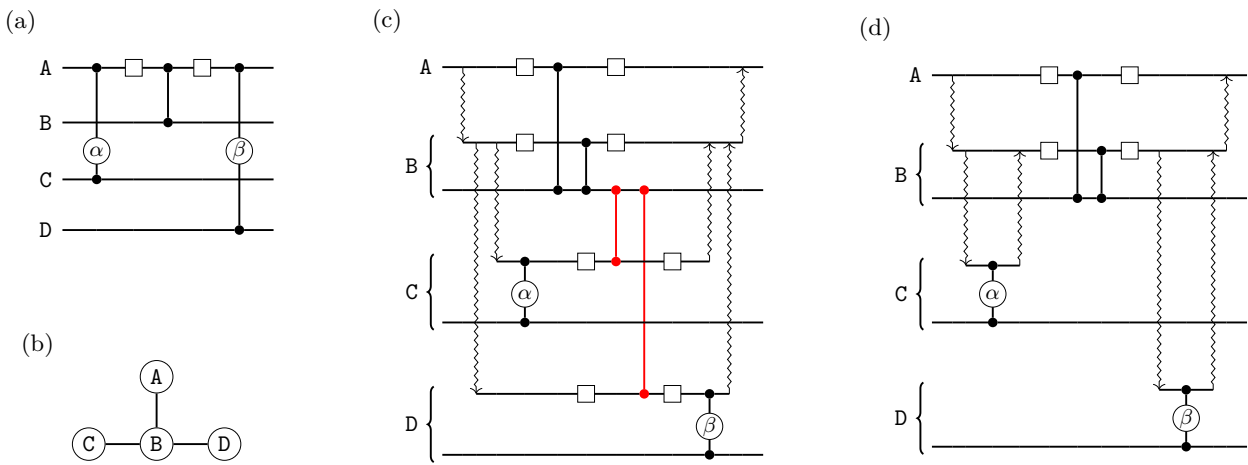

    \centering
    \begin{tikzpicture}
    \node[scale=0.9] (input) {\input{images/ALAP/input}};
    \node[below=13mm of input, scale=0.9] (network) {
        \begin{tikzpicture}
            \node[draw, circle, thin, minimum size=1.2em, inner sep=2pt] (A) {\texttt{A}};
            \node[below=5mm of A, draw, circle, thin, minimum size=1.2em, inner sep=2pt] (B) {\texttt{B}};            
            \node[left=5mm of B, draw, circle, thin, minimum size=1.2em, inner sep=2pt] (C) {\texttt{C}};
            \node[right=5mm of B, draw, circle, thin, minimum size=1.2em, inner sep=2pt] (D) {\texttt{D}};
            \draw[thick] (A) -- (B);
            \draw[thick] (C) -- (B);
            \draw[thick] (D) -- (B);
        \end{tikzpicture}
    };
    \node[below right=-28mm and 13mm of input, scale=0.9] (wrong) {\input{images/ALAP/wrong}};
    \node[right=13mm of wrong, scale=0.9] (right) {\input{images/ALAP/right}};
    \node[above left=0mm and -3mm of input] {\small (a)};
    \node[above left=-2mm and 0mm of network] {\small (b)};
    \node[above left=0mm and -3mm of wrong] {\small (c)};
    \node[above left=0mm and -3mm of right] {\small (d)};
\end{tikzpicture}
    \caption{\textbf{Combining embedding and Steiner trees.} (a) An input circuit where each qubit is allocated to a different module, (b) the topology of the network of modules. (c) An equivalent circuit generated using the approach from \cref{sec:Steiner}: two non-local correction gates arise, drawn in red. (d) An equivalent circuit generated by Algorithm~\ref{alg:ALAP}.}
    \label{fig:ALAP}
\end{figure*}

There is a simple solution to our compatibility issue: maintain the entanglement of these proxy link qubits for as long as possible to maximise the use of Steiner trees, but disentangle them right before an embedding unit so that they do not interfere with it. The implementation of such an intuition is sketched in Algorithm~\ref{alg:ALAP}.

\begin{figure}
    \begin{algorithm}[H]
        \caption{Distribution with embedding and Steiner trees}
        \label{alg:ALAP}
        \textbf{Input:} hyperedge (\textsf{hedge}), allocation map ($\phi$)
        
        \textbf{Output:} partly distributed circuit (\textsf{dist\_circ}).

        \hrulefill
        \begin{algorithmic}[1]
            \State \textsf{tree} $\leftarrow$ \textsf{hedge}'s Steiner tree (see \cref{sec:Steiner})
            \State \textsf{linked\_modules} $\leftarrow \varnothing$ 
            \State \textsf{hedge\_circ} $\leftarrow$ extract as in Remark~\ref{rmk:ALAP}
            \State \textsf{iter} $\leftarrow$ \textsf{hedge\_circ.iterator()}
            \State 
            \State \textsf{dist\_circ} $\leftarrow$ empty circuit
            \While{\textsf{iter.current} \textbf{not} \textsf{null}}
                \State \textsf{gate} $\leftarrow$ \textsf{iter.current}
                \State
                \If{\textsf{gate} $\in$ \textsf{hedge}} {\color{gray}\Comment{Distribute}}
                    \If{$\phi(\textsf{gate}) \not\in \textsf{linked\_modules}$}
                        \State \textsf{dist\_circ} $\leftarrow$ insert starting process
                        \State \textsf{linked\_modules} $\leftarrow$ add $\phi(\textsf{gate})$
                    \EndIf
                    \State \textsf{dist\_circ} $\leftarrow$ insert distributed \textsf{gate}
                \State
                \ElsIf{\textsf{gate} is \H{}} {\color{gray}\Comment{Embed}}
                    \State \textsf{embedding\_unit} $\leftarrow$ \textsf{gate}'s embedding unit
                    \State \textsf{remote\_module} $\leftarrow$ module $\texttt{B}$ from Lemma~\ref{lem:embedding_cond}
                    \For{\textsf{module} $\in$ \textsf{linked\_modules}}
                        \If{\textsf{module} $\not =$ \textsf{remote\_module}}
                            \State \textsf{dist\_circ} $\leftarrow$ insert ending process
                        \EndIf
                    \EndFor
                    \State \textsf{linked\_modules} $\leftarrow$ \textsf{\{remote\_module\}}
                    \State \textsf{dist\_circ} $\leftarrow$ insert \textsf{embedding\_unit}
                    \State \textsf{dist\_circ} $\leftarrow$ insert correction gates
                    \State \textsf{iter} $\leftarrow$ move at the end of \textsf{embedding\_unit}
                \State
                \Else {\color{gray}\Comment{Skip}}
                    \State \textsf{dist\_circ} $\leftarrow$ insert \textsf{gate} unchanged
                \EndIf
                \textsf{iter.next()}
            \EndWhile
        \end{algorithmic}
    \end{algorithm}
\end{figure}

\cref{fig:ALAP}d shows the result of running Algorithm~\ref{alg:ALAP} on a simple circuit. The proxy link qubit of module \texttt{B} is maintained throughout the circuit, whereas the link qubits of modules $\texttt{C}$ and $\texttt{D}$ are only maintained as long as necessary to implement the two \CU{} gates. Maintaining the link qubit of module \texttt{B} saves one ebit, whereas our management of the link qubits of modules $\texttt{C}$ and $\texttt{D}$ avoids the need for non-local correction gates that would otherwise be required (see \cref{fig:ALAP}c). Thus, it is possible to define distributions that combine the techniques of embedding and Steiner trees, and Algorithm~\ref{alg:ALAP} is capable of generating the corresponding circuit. 

We can count the number of ebits consumed in the distributed circuit outputted by Algorithm~\ref{alg:ALAP}, thus obtaining the exact ebit cost of the distribution. This can be done for each cut hyperedge in our hypergraph, and it is straightforward to check that Algorithm~\ref{alg:ALAP} runs in time $\bigO{g_d + g_e}$ where $g_d$ is the number of gate-vertices in the hyperedge and $g_e$ is the number of gates that need to be embedded to realise its distribution. Thus, this provides an efficient function to calculate the exact ebit cost of a given cut hyperedge, using both embedding and Steiner trees. This cost function will be used by the combinatorial optimisation approaches of \cref{sec:refinement} which will be the ones to ultimately decide how each non-local gates should be distributed.

\begin{remark} \label{rmk:ALAP}
    Algorithm~\ref{alg:ALAP} iterates over the hyperedge's subcircuit (\texttt{hedge\_circ}): given a hyperedge whose qubit-vertex is $\hat{q}$, its subcircuit is the sequence of gates from the original circuit that contains all of the gates corresponding to gate-vertices of the hyperedge and every gate in between these that acts on $\hat{q}$. 
    The hyperedge given to Algorithm~\ref{alg:ALAP} as input is required to be valid, in the sense that every gate in its subcircuit is either distributable or embeddable. We can verify this ahead of time by checking the conditions from Lemma~\ref{lem:distributable_cond} (with the amend from Remark~\ref{rmk:distributable_cond}) and Lemma~\ref{lem:embedding_cond} respectively.
\end{remark}


\subsection{Partitioning on heterogenous networks}
\label{sec:hypergraph_heterogeneous}

In \cref{sec:hypergraph} we reviewed an approach that reduces the DQC problem on fully connected networks to hypergraph partitioning~\cite{Andr_s_Mart_nez_2019}. In the case of heterogeneous networks, the DQC problem still reduces to (a version of) hypergraph partitioning, but the cost function of a partition is different --- since we need to consider the distance between modules --- and we must filter out invalid solutions where the module's computation register capacity is exceeded. In this section we propose two approaches to solve this alternative version of hypergraph partitioning and, thus, the DQC problem on heterogeneous networks. 

Both of our approaches start from an initial partition and apply rounds of updates to it, guided by the cost function defined in \cref{sec:ALAP}. On each round, vertices of the hypergraph are moved from their assigned module to a different one; then, the cost of every hyperedge containing a reallocated vertex is updated. We can calculate the gain of the moves as the difference between the new cost and the previous cost. Depending on the gain and the approach used, the moves will be committed or rolled back. Since calculation of the cost function from \cref{sec:ALAP} requires finding Steiner trees on the network's graph --- which is a non-trivial computation --- we keep a cache of already computed Steiner trees. 

Recall that our hypergraphs have two kinds of vertices: qubit-vertices and gate-vertices. The allocation of a qubit-vertex to a module fills up one slot of the module's computation register, whereas the allocation of gate-vertices do not affect the computation register. Consequently, we assign weight $1$ to qubit-vertices and weight $0$ to gate-vertices and filter out partitions where the sum of weights in a module exceeds the corresponding module's computation register capacity. If a move would cause the capacity of a module to be exceeded, we select a qubit-vertex on the offending module and swap it with the vertex we intended to move. Our approaches assume unbounded link qubit registers, unlike~\cite{Sundaram2022GeneralDistribution}. In contrast, we make use of Steiner trees as discussed in~\cref{sec:Steiner}, tapping into optimisation opportunities not considered in the latter work.

\subsubsection{Simulated annealing}
\label{sec:annealing}

Simulated annealing is a stochastic optimisation algorithm; modifying an existing solution by randomly searching its neighbourhood. This search process is repeated iteratively, with the working solution updated if a lower cost solution is found. The solution may also be updated with some probability if the cost is higher, which prevents becoming trapped in local optima. The probability of accepting a worse solution falls with each iteration, encouraging that the region of the global optimum be found early on, after which the optimum itself is isolated.

In particular, the initial circuit distribution we use assigns qubits to random modules which have space for them, and assigns gate-vertices to random modules as well. Each step moves a random vertex in the distribution hypergraph to a random module. In the case of qubit-vertices this may require that a qubit in the module be swapped out to make room. The new distribution is accordingly updated depending on the new cost of the distribution.

Each iterations of the annealing procedure makes use of the cost function defined in \cref{sec:ALAP} to accept or reject an update to the distribution. As such, the scheme considers heterogeneous networks and Steiner trees in the first instance. It will not however update the distributable packets, and so considers embedding only in so far as the initial distribution take it into account. In \cref{sec:benchmarks} the initial distribution does not take embedding into consideration. Since annealing is a very general purpose tool and not well optimised to the problem of concern, we do not expect it to perform as well or as quickly as other specialised tool. The technique is however very versatile, and could be easily adapted to other similar problem. Additionally our implementation in \pytketdqc{} avoids dependencies on other third party libraries.

\subsubsection{Boundary reallocation} 
\label{sec:boundary_realloc}

The initial solution of this approach is computed using \KaHyPar{}~\cite{KaHyPar}, a state-of-art hypergraph partitioner that has the option to fix the maximum vertex weight each partition block can hold. Thus, its solution already provides a valid distribution, in the sense that it does not exceed the computation register capacity of the modules. However, the solution is optimised according to the wrong cost function, since it is assuming an all-to-all network topology. We refine the solution applying a greedy algorithm guided by the cost function defined in \cref{sec:ALAP}, improving the allocation of vertices on the boundary between partition blocks. 

On each round, we collect all of the vertices in the hypergraph that belong to a hyperedge cut by the partition --- the boundary of the partition. For each vertex $v$ in said boundary we find all of the modules that $v$ has a neighbour in; we then calculate the gain of moving $v$ to each of these modules and pick the most advantageous move (with ties broken randomly) or, if all of them are detrimental, we choose not to move $v$. A round finishes when this routine has been run once for each vertex in the boundary. Thus, each round generates a new partition and the cost of its distribution is decreased monotonically.

There is no attempt to escape local minima. The initial solution provided by \KaHyPar{} --- which does have strategies to avoid local minima~\cite{KaHyPar} --- already identifies groups of qubits that should be allocated to the same module; such grouping is a property of the circuit and hence, equally valid in the context of heterogeneous networks. Unfortunately, our greedy refinement struggles to move vertices that have many neighbours within its allocated module but few in other modules. We expect this to be a noticeable limitation in the case of networks resembling a line graph, where some of these immobile vertices may be stuck on the ends of the network. In practice, however, we expect that modules will be arranged in a small-world network\footnote{In a small-world network of $N$ nodes, few of them are adjacent to each other, but the path between any two nodes tends to be of length $\log N$. Small-world networks are common in engineering due to their logarithmic scaling average distance, which reduces communication bottlenecks~\cite{SmallWorldNetEfficiency}.} such as a hypercube, where the allocation of a few immobile vertices is not crucial thanks to the network's small average distance. In such cases, the potential for optimisation would primarily come from making smart choices of where the vertices that do not strongly belong to any of the modules (\ie{} those in the boundary of the partition), taking into account the topology of the network.

\section{Benchmarks}
\label{sec:benchmarks}

Here we present the results of benchmarking the methods described in \cref{sec:solutions}, comparing them to \cite{gsundaram_et_al}. We describe the networks, circuits, and distribution workflows used in \cref{sec:benchmarks networks}, \cref{sec:benchmarks circuits} and \cref{sec:benchmarks workflows} respectively. The results of the benchmarks are shown and discussed in \cref{sec:benchmarks results}.

\subsection{Networks}
\label{sec:benchmarks networks}

The following architectures are used in the experiments of \cref{sec:benchmarks results}. Generator methods for these networks are available within \pytketdqc{}. 
\begin{description}
    \item[Homogeneous:] All modules are directly connected to all other modules. All modules contain the same number of qubits, and no bound is set on the number of link qubits available in each module. This models an idealised network, and is exemplified in \cref{fig:benchmarks networks homogeneous}.
\end{description}
We refer to the following collectively as \emph{heterogeneous networks}. We will generate random instances of heterogeneous networks, and they are designed to be representative of real world networks.
\begin{description}
    \item[Unstructured:] Modules are connected according to edges in random Erd\'os–R\'enyi graphs, where each possible edge in the graph is added with a fixed probability. In our case we post-select to generate only connected graphs. This is the most common notion of random networks, and is exemplified in \cref{fig:benchmarks networks random}.
    
    \item[Scale-free:] The distribution of node degrees in a scale-free network follows a power law. Such networks have few nodes, called hubs, with high degree. This is a common model for networks, including the World Wide Web \cite{Baraba_si_1999}. They can be generated using preferential attachment, where high degree nodes are more likely to receive new edges as nodes are added. This is the case for the Barab\'asi–Albert model \cite{Baraba_si_1999} of scale-free networks, which we use to generate them here.\footnote{We find this broad class of networks to be a well motivated example for the purposes of our comparison. However, practical considerations give subdivisions of the class of scale-free networks \cite{doyle2005robust}. A fine grained analysis of the resulting impact on quantum circuit distribution would be of interest.} Scale-free networks are exemplified in \cref{fig:benchmarks networks scale free}.
    
    \item[Small-world:] The characteristic path lengths of small-world networks are small, while the clustering coefficient is large \cite{Watts1998, SmallWorldNetEfficiency}. This is compared to random Erd\'os–R\'enyi graphs which have small characteristic path and small clustering coefficient. Unlike Scale-free networks, small-world networks do not include hub nodes. Such networks are used to model social networks and are prevalent in engineering due to their communication efficiency~\cite{SmallWorldNetEfficiency}. We generate them using the Watts–Strogatz model \cite{Watts1998}, and exemplify them in \cref{fig:benchmarks networks small world}.
\end{description}

The particular sizes of the networks we use are listed in the results of \cref{sec:benchmarks results}. In the case of the heterogeneous networks, edge probabilities are set so that the average number of edges incident on each module is two, and qubits are assigned at random to each module. We take that the size of the link qubit register is the largest integer smaller than the average number of computational qubits per module. This means that one would not typically be able to fit the computational qubits of one network module into the link qubit register of another, and as such that networking the modules together results in an increase in the number of computation qubits. Bounds to the size of the link qubit register are not considered in \cref{sec:benchmarks results}, but are explored in \cref{sec:limited links}.

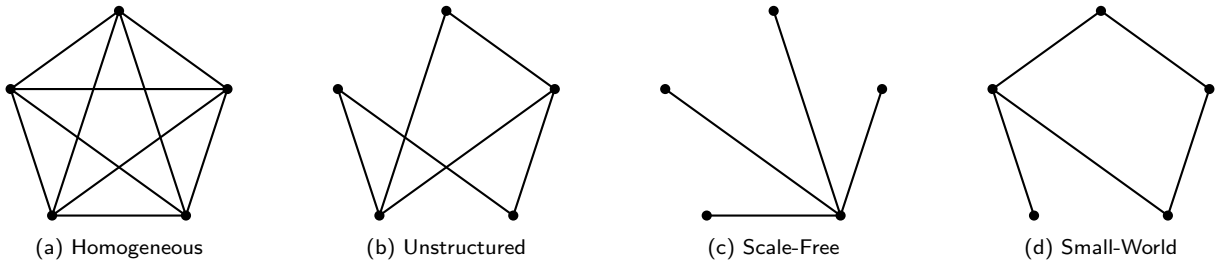
\begin{figure*}
    \centering
    \begin{subfigure}[b]{0.24\textwidth}
        \centering
        \begin{tikzpicture}[thick, scale=1.5]
            \coordinate (a) at (0.587,-0.809);
            \coordinate (b) at (0.951,0.309);
            \coordinate (c) at (0,1);
            \coordinate (d) at (-0.951,0.309);
            \coordinate (e) at (-0.587,-0.809);

            \filldraw[black] (a) circle (1pt);
            \filldraw[black] (b) circle (1pt);
            \filldraw[black] (c) circle (1pt);
            \filldraw[black] (d) circle (1pt);
            \filldraw[black] (e) circle (1pt);
            
            \draw (a) -- (b);
            \draw (a) -- (c);
            \draw (a) -- (d);
            \draw (a) -- (e);

            \draw (b) -- (c);
            \draw (b) -- (d);
            \draw (b) -- (e);

            \draw (c) -- (d);
            \draw (c) -- (e);

            \draw (d) -- (e);
        \end{tikzpicture}
        \caption{Homogeneous}
        \label{fig:benchmarks networks homogeneous}
    \end{subfigure}
    \hfill
    \begin{subfigure}[b]{0.24\textwidth}
        \centering
        \begin{tikzpicture}[thick, scale=1.5]
            \coordinate (a) at (0.587,-0.809);
            \coordinate (b) at (0.951,0.309);
            \coordinate (c) at (0,1);
            \coordinate (d) at (-0.951,0.309);
            \coordinate (e) at (-0.587,-0.809);

            \filldraw[black] (a) circle (1pt);
            \filldraw[black] (b) circle (1pt);
            \filldraw[black] (c) circle (1pt);
            \filldraw[black] (d) circle (1pt);
            \filldraw[black] (e) circle (1pt);
            
            \draw (a) -- (b);
            \draw (a) -- (d);

            \draw (b) -- (c);
            \draw (b) -- (e);

            \draw (c) -- (e);

            \draw (d) -- (e);
        \end{tikzpicture}
        \caption{Unstructured}
        \label{fig:benchmarks networks random}
    \end{subfigure}
    \hfill
    \begin{subfigure}[b]{0.24\textwidth}
        \centering
        \begin{tikzpicture}[thick, scale=1.5]
            \coordinate (a) at (0.587,-0.809);
            \coordinate (b) at (0.951,0.309);
            \coordinate (c) at (0,1);
            \coordinate (d) at (-0.951,0.309);
            \coordinate (e) at (-0.587,-0.809);

            \filldraw[black] (a) circle (1pt);
            \filldraw[black] (b) circle (1pt);
            \filldraw[black] (c) circle (1pt);
            \filldraw[black] (d) circle (1pt);
            \filldraw[black] (e) circle (1pt);
            
            \draw (a) -- (b);
            \draw (a) -- (c);
            \draw (a) -- (d);
            \draw (a) -- (e);



        \end{tikzpicture}
        \caption{Scale-Free}
        \label{fig:benchmarks networks scale free}
    \end{subfigure}
    \hfill
    \begin{subfigure}[b]{0.24\textwidth}
        \centering
        \begin{tikzpicture}[thick, scale=1.5]
            \coordinate (a) at (0.587,-0.809);
            \coordinate (b) at (0.951,0.309);
            \coordinate (c) at (0,1);
            \coordinate (d) at (-0.951,0.309);
            \coordinate (e) at (-0.587,-0.809);

            \filldraw[black] (a) circle (1pt);
            \filldraw[black] (b) circle (1pt);
            \filldraw[black] (c) circle (1pt);
            \filldraw[black] (d) circle (1pt);
            \filldraw[black] (e) circle (1pt);
            
            \draw (a) -- (b);
            \draw (a) -- (d);

            \draw (b) -- (c);

            \draw (c) -- (d);

            \draw (d) -- (e);
        \end{tikzpicture}
        \caption{Small-World}
        \label{fig:benchmarks networks small world}
    \end{subfigure}
    \caption{\textbf{Example network architecture graphs.} Vertices indicate modules. Edges indicate connections along which ebits can be established.}
    \label{fig:benchmarks networks}
\end{figure*}

\subsection{Circuits}
\label{sec:benchmarks circuits}

The following classes of randomly generated circuits are considered during the experiments of \cref{sec:benchmarks results}.
\begin{description}
    \item[CZ Fraction:] Consisting of $d$ layers of gates, with each layer built from \H{} and \CZ{} gates. A parameter \texttt{cz\_fraction} determines the proportion of the qubits on which \CZ{} gates are acted in each layer. These benchmark circuits are already in the gateset considered by the distribution workflows studied, and so provide a controlled way to study the performance of these workflows. \CZ{} fraction circuits are introduced in \cite{gsundaram_et_al}, exemplified in \cref{fig:benchmarks circuits cz fraction}, and detailed in \cref{alg:benchmark circuits cz fraction}.
\end{description}
While CZ Fraction circuits were designed for the study of DQC workflows, the following are inspired by popular protocols.
\begin{description}
    \item[Quantum Volume:] Consists of $d$ \emph{layers} of random two-qubit gates, each acting on different bipartitions of the qubits, and similar to those used for the quantum volume benchmark \cite{Cross_2019}. By utilising uniformly random two-qubit unitaries and all-to-all connectivity, Quantum Volume Circuits provide a comprehensive benchmark. While CZ Fraction and Pauli Gadget circuits naturally decompose to contain \CZ{} gates when rewritten in \gateset{}, Quantum Volume circuits will contain \CU{} gates of a variety of rotation angles. This exemplifies the capacity for \pytketdqc{} to distribute such gates. Quantum Volume circuits are exemplified in \cref{fig:benchmarks circuits random} and detailed in \cref{alg:benchmark circuits random}. 
    
    \item[Pauli Gadget:] Pauli gadgets \cite{Cowtan_2020} are quantum circuits implementing the exponential of a Pauli tensor. Sequences of Pauli gadgets acting on qubits form \emph{product formula} circuits, most commonly used in Hamiltonian simulation and the variational quantum eigensolver (VQE)\cite{Berry_2006, Peruzzo_2014, barkoutsos2018quantum}. Circuits from this particular class of Pauli Gadget circuits are constructed from several layers of random Pauli Gadgets, each acting on a random subset of $n$ qubits \cite{Mills_2021}. Pauli Gadget circuits are exemplified in \cref{fig:benchmarks circuits pauli gadget} and detailed in \cref{alg:benchmark circuits pauli}.
\end{description}
In the case of all benchmarks conducted in this work, the number of layers used is set to be equal to the number of qubits in the circuit.

The comparative size of the circuits in these classes is seen in \cref{fig:benchmarks circuits size}. Note that \CZ{} fraction circuits contain many fewer two-qubit gates than circuits from the other two classes. This is because each layer of the Quantum Volume and Pauli Gadget circuits corresponds to many gates when decomposed into the \gateset{} gate set. Further, while circuits spanning the same number of qubits in the Quantum Volume class contain more two-qubit gates than those in the Pauli Gadget class, this number is comparable. 

\begin{figure}
	\begin{algorithm}[H]
		\caption{Building an instance of CZ Fraction.}
		\label{alg:benchmark circuits cz fraction}
		\inout{Width, $n \in \mathbb{Z}$, depth, $d \in \mathbb{Z}$, fraction $p \in \left[ 0 , 1 \right]$}{Circuit, $C_n$}
		\begin{algorithmic}[1]
			\For{each layer $t$ up to depth $d$}
                    \For{each qubit $q_i$}
			             \State With probability $1-p$ apply \H{}.
                    \EndFor
                    \State Randomly pair all qubits to which no \H{} was acted.
                    \State To each pair apply \CZ{}.
			\EndFor
		\end{algorithmic}
	\end{algorithm}
	\begin{algorithm}[H]
		\caption{Building an instance of Quantum Volume.}
		\label{alg:benchmark circuits random}
		\inout{Width, $n \in \mathbb{Z}$, depth, $d \in \mathbb{Z}$}{Circuit, $C_n$}
		\begin{algorithmic}[1]
			\For{each layer $t$ up to depth $d$} 
			\State Divide qubits into $\frac{n}{2}$ random pairs $\left\{q_{i,1}, q_{i,2}\right\}$.
			\ForAll{$i \in \mathbb{Z}$, $0 \leq i \leq \frac{n}{2}$}
			\State Generate $U_{i,t} \in \mathrm{SU} \brac{4}$ uniformly at random according to the Haar measure.
			\State Enact the gate corresponding to the unitary $U_{i,t}$ on qubits $q_{i, 1}$ and $q_{i,2}$. \Comment Decompositions of this gate can be found in \cite{tucci2005introduction}
			\EndFor
			\EndFor
		\end{algorithmic}
	\end{algorithm}
	\begin{algorithm}[H]
		\caption{Building an instance of Pauli Gadget.}
		\label{alg:benchmark circuits pauli}
		\inout{Width, $n \in \mathbb{Z}$, depth, $d \in \mathbb{Z}$}{Circuit, $C_n$}
		\begin{algorithmic}[1]
			\For{each layer $t$ up to depth $d$}
			\State Select a random string $s^t \in \left\{I,X,Y,Z \right\}^n$
			\State Generate random angle $\alpha^t \in \left[0 , 2 \pi \right]$
			\State Enact $\exp \brac{i\bigotimes_j s^t_j \alpha^t}$ on qubits $q_1 , ... , q_n$. \Comment Decompositions of this gate can be found in \cite{Cowtan_2020}
			\EndFor
		\end{algorithmic}
	\end{algorithm}
\end{figure}

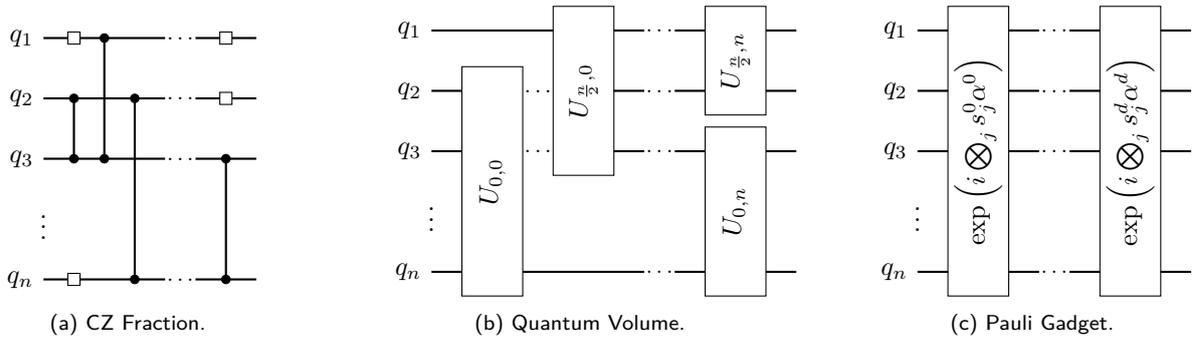
\begin{figure*}
\begin{subfigure}[t]{0.3\textwidth}
        \centering
        \begin{tikzpicture}[scale=0.8]
        \draw[thick] (-0.5,1) node[anchor=east] {$q_n$} -- (1.5,1);
        \node at (-0.5, 2) {\vdots};
		\draw[thick] (-0.5,3) node[anchor=east] {$q_3$} -- (1.5,3);
		\draw[thick] (-0.5,4) node[anchor=east] {$q_2$} -- (1.5,4);
		\draw[thick] (-0.5,5) node[anchor=east] {$q_1$} -- (1.5,5);

            \filldraw[black] (0, 4) circle (2pt);
            \draw[thick] (0, 4) -- (0, 3);
            \filldraw[black] (0, 3) circle (2pt);

            \filldraw[black] (0.5, 5) circle (2pt);
            \draw[thick] (0.5, 5) -- (0.5, 3);
            \filldraw[black] (0.5, 3) circle (2pt);

            \filldraw[black] (1, 4) circle (2pt);
            \draw[thick] (1, 4) -- (1, 1);
            \filldraw[black] (1, 1) circle (2pt);

            \draw[fill=white] (-0.1,0.9) rectangle (0.1,1.1);
            \draw[fill=white] (-0.1,4.9) rectangle (0.1,5.1);

        \path (1.5,5) -- node {\footnotesize $\,\dots$} (2,5);
		\path (1.5,4) -- node {\footnotesize $\,\dots$} (2,4);
		\path (1.5,3) -- node {\footnotesize $\,\dots$} (2,3);
		\path (1.5,1) -- node {\footnotesize $\,\dots$} (2,1);

        \draw[thick] (2,5) -- (3, 5);
		\draw[thick] (2,4) -- (3, 4);
		\draw[thick] (2,3) -- (3, 3);
		\draw[thick] (2,1) -- (3, 1);

            \filldraw[black] (2.5, 3) circle (2pt);
            \draw[thick] (2.5, 3) -- (2.5, 1);
            \filldraw[black] (2.5, 1) circle (2pt);

            \draw[fill=white] (2.4,3.9) rectangle (2.6,4.1);
            \draw[fill=white] (2.4,4.9) rectangle (2.6,5.1);
        \end{tikzpicture}
        \caption{CZ Fraction.}
        \label{fig:benchmarks circuits cz fraction}
    \end{subfigure}
    \hfill
    \begin{subfigure}[t]{0.3\textwidth}
        \centering
	\begin{tikzpicture}[scale=0.8]
        \draw[thick] (-0.5,1) node[anchor=east] {$q_n$} -- (0,1);
        \node at (-0.5, 2) {\vdots};
		\draw[thick] (-0.5,3) node[anchor=east] {$q_3$} -- (0,3);
		\draw[thick] (-0.5,4) node[anchor=east] {$q_2$} -- (0,4);
		\draw[thick] (-0.5,5) node[anchor=east] {$q_1$} -- (1.5,5);

		\draw (0,0.6) rectangle (1,4.4) node[pos=.5, rotate=90] {$U_{0,0}$};
		\draw (1.5,2.6) rectangle (2.5,5.4) node[pos=.5, rotate=90]
        {$U_{\frac{n}{2},0}$};
		
		\path (1,3) -- node {\footnotesize $\,\dots$} (1.5,3);
		\path (1,4) -- node {\footnotesize $\,\dots$} (1.5,4);

		\draw[thick] (2.5,5) -- (3,5);
		\draw[thick] (2.5,4) -- (3,4);
		\draw[thick] (2.5,3) -- (3,3);
		\draw[thick] (1,1) -- (3,1);

        \path (3,5) -- node {\footnotesize $\,\dots$} (3.5,5);
		\path (3,4) -- node {\footnotesize $\,\dots$} (3.5,4);
		\path (3,3) -- node {\footnotesize $\,\dots$} (3.5,3);
		\path (3,1) -- node {\footnotesize $\,\dots$} (3.5,1);

        \draw[thick] (3.5,5) -- (4, 5);
		\draw[thick] (3.5,4) -- (4, 4);
		\draw[thick] (3.5,3) -- (4, 3);
		\draw[thick] (3.5,1) -- (4, 1);
		
		\draw (4,0.6) rectangle (5,3.4) node[pos=.5, rotate=90]{$U_{0,n}$};
		\draw (4,3.6) rectangle (5,5.4) node[pos=.5, rotate=90]{$U_{\frac{n}{2},n}$};

        \draw[thick] (5,5) -- (5.5, 5);
		\draw[thick] (5,4) -- (5.5, 4);
		\draw[thick] (5,3) -- (5.5, 3);
		\draw[thick] (5,1) -- (5.5, 1);
		
	\end{tikzpicture}
    \caption{Quantum Volume.}
    \label{fig:benchmarks circuits  random}
    \end{subfigure}
    \hfill
    \begin{subfigure}[t]{0.3\textwidth}
        \centering
	\begin{tikzpicture}[scale =0.8]
		\draw[thick] (-0.5,1) node[anchor=east] {$q_n$} -- (0,1);
		\node at (-0.5, 2) {\vdots};
		\draw[thick] (-0.5,3) node[anchor=east] {$q_3$} -- (0,3);
		\draw[thick] (-0.5,4) node[anchor=east] {$q_2$} -- (0,4);
		\draw[thick] (-0.5,5) node[anchor=east] {$q_1$} -- (0,5);
		
		\draw (0,0.6) rectangle (1,5.4) node[pos=.5, rotate=90] {$\exp \brac{i\bigotimes_j s^0_j \alpha^0}$};
		
		\draw[thick] (1,1) -- (1.5,1);
		\draw[thick] (1,3) -- (1.5,3);
		\draw[thick] (1,4) -- (1.5,4);
		\draw[thick] (1,5) -- (1.5,5);
		
		\path (1.5,1) -- node {\footnotesize $\,\dots$} (2,1);
		\path (1.5,3) -- node {\footnotesize $\,\dots$} (2,3);
		\path (1.5,4) -- node {\footnotesize $\,\dots$} (2,4);
		\path (1.5,5) -- node {\footnotesize $\,\dots$} (2,5);
		
		\draw[thick] (2,1) -- (2.5,1);
		\draw[thick] (2,3) -- (2.5,3);
		\draw[thick] (2,4) -- (2.5,4);
		\draw[thick] (2,5) -- (2.5,5);
		
		\draw (2.5,0.6) rectangle (3.5,5.4) node[pos=.5, rotate=90] {$\exp \brac{i\bigotimes_j s^d_j \alpha^d}$};
		
		\draw[thick] (3.5,1) -- (4,1);
		\draw[thick] (3.5,3) -- (4,3);
		\draw[thick] (3.5,4) -- (4,4);
		\draw[thick] (3.5,5) -- (4,5);
		
	\end{tikzpicture}
        \caption{Pauli Gadget.}
        \label{fig:benchmarks circuits pauli gadget}
    \end{subfigure}
    \caption{\textbf{Examples of circuits used for benchmarks.}}
    \label{fig:benchmarks circuits}
\end{figure*}

\begin{figure}
    \centering
     \includegraphics[width=\columnwidth]{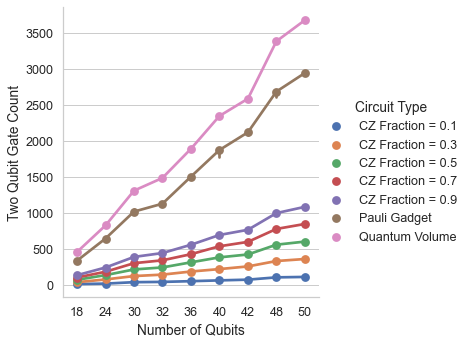}
     \caption{\textbf{Average number of two-qubit gates for circuit of each type.} Points indicate the mean number of two-qubit gates for circuits of the given type, covering the given number of qubits. Error bars indicate one standard deviation.}
     \label{fig:benchmarks circuits size}
\end{figure}

\subsection{Distribution workflows}
\label{sec:benchmarks workflows}




This section details the distribution workflows used in the experiments of \cref{sec:benchmarks results}. Our novel distribution workflows improve upon the distributions output by the following schemes, presented in the literature~\cite{Junyi2022, Andr_s_Mart_nez_2019} and available through \pytketdqc{}
\begin{description}
    \item[\ce{}:] Utilises the approach discussed in \cref{sec:vertex_cover_embedding} for distributing quantum circuits using vertex covering.
    \item[\hp{}:] Utilises the approach discussed in \cref{sec:hypergraph} for distributing quantum circuits using hypergraph partitioning.
\end{description}
The following workflows are novel to this work, and are available through \pytketdqc{}. The refinement passes referenced here are detailed further in \cref{sec:refinement}.
\begin{description}
    \item[\ces{}:] All gates in each hyperedge of distributions resulting from \ce{} act between the same two modules. \ces{} improves upon the output of \ce{} by merging packets where doing so does not require additional embedding, as discussed in \cref{sec:refinement steiner}. This results in an ebit saving from reusing proxy link qubits when distributing entanglement according Steiner trees.
    \item[\cesd{}:] Non-local gates are allocated by \ce{} to either one of the two modules that contain the qubits the gate acts on. Gates are not reallocated by \ces{}. \cesd{} improves upon the latter by reallocating gates, making use of detached gates to save extra ebits (details in \cref{sec:refinement detatched}). Note that this improvement upon \ce{} is made possible by first refining by merging hyperedges, as in \ces{}, as detached gates may be beneficially utilised when hyperedges contain gates acting between 3 or more modules.
    \item[\pe{}:] Refines the approach of \hp{} to make use of embedding (see \cref{sec:refinement embedding}). This does not consider heterogeneous network connectivity, and we will only use it on homogeneous networks.
    \item[\ph{}:] Recreates the approach of \cref{sec:boundary_realloc} to adapt the output of \hp{} to heterogeneous networks using boundary reallocation.
    \item[\phe{}:] Since \ph{} neglects the possibility of embedding gates, \phe{} improves upon it by making use of embedding to merge distributable packets, as discussed in \cref{sec:refinement embedding}.
    \item[\anneal{}:] Utilises the approach of \cref{sec:annealing} for quantum circuit distribution using simulated annealing. \anneal{} optimises for heterogeneous networks in the first instance, including detached gates and Steiner trees, but does not consider embedding.
\end{description}
The following workflows correspond to existing approaches \cite{gsundaram_et_al} discussed in \cref{sec:vertex_cover}. The necessary implementations are not available in \pytketdqc{} but were provided by their authors upon request. These approaches perform qubit allocation by solving a balanced k-min-cut problem over an edge-weighted graph, where the weights capture the connectivity of the circuit. A greedy algorithm that iteratively fixes allocations of non-local gates is used. Each iteration requires solving an instance of the weighted densest subgraph problem in order to pick the allocations to fix at that round. The following two distribution workflows use different methods of solving the weighted densest subgraph problem.
\begin{description}
    \item[\gss{}:] A simple greedy solution.
    \item[\gslp{}:] An optimal approach based on Integer Linear Programming.
\end{description}

\subsection{Results}
\label{sec:benchmarks results}

Here we present the results of the benchmarks described above. We explore homogeneous networks in \cref{sec:benchmarks results homogeneous}, heterogeneous networks in \cref{sec:benchmarks results heterogeneous}, and the distribution of a particular circuit of practical interest over heterogeneous networks in \cref{sec:benchmarks results chemistry}.

\subsubsection{Homogeneous networks}
\label{sec:benchmarks results homogeneous}

We compare the techniques described in \cref{sec:solutions} to the techniques of \cite{gsundaram_et_al}, namely \gss{} and \gslp{}. Aligning with the target scenario of \cite{gsundaram_et_al}, we consider homogeneous networks and CZ Fraction circuits. We consider networks with 4, 5 and 6 modules, each with 8 qubits per module, as well as 2 module networks with 16 and 25 qubits per module. For each network size we generate 5 random CZ fraction circuits of that size.

Results concerning networks with more than 2 modules can be seen in \cref{fig:benchmarks homogeneous cz fraction}. Consistently, the unrefined distribution workflows producing the lowest cost distributions are \hp{} and \gss{}.\footnote{Note that this contrasts with the results reported in \cite{gsundaram_et_al}. This is the result of correcting a poor choice of default parameters in \cite{Andr_s_Mart_nez_2019}, which limited how large a hyperedge could be.} For smaller networks \hp{} mildly outperforms \gss{}.

\begin{figure*}
    \centering
     \begin{subfigure}[b]{\textwidth}
         \centering
         \includegraphics[width=\textwidth]{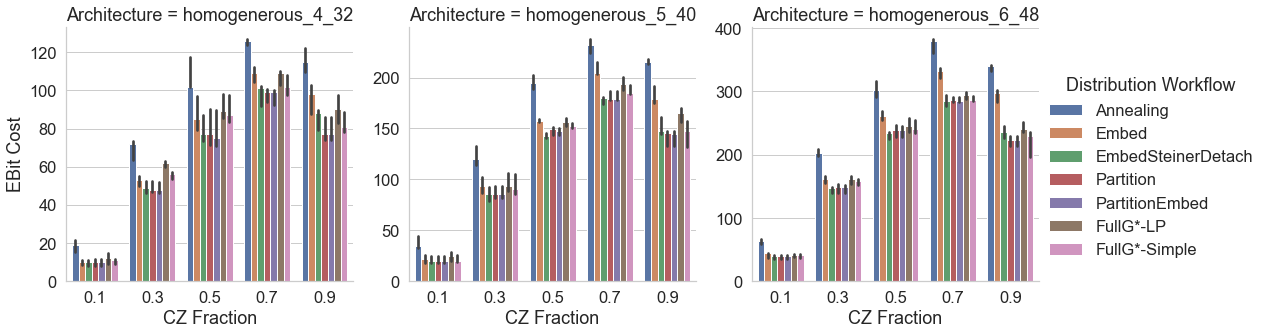}
         \caption{Ebit cost.}
         \label{fig:benchmarks homogeneous cz fraction cost}
     \end{subfigure}
     
     \vspace{10pt}
     
     \begin{subfigure}[b]{\textwidth}
         \centering
         \includegraphics[width=\textwidth]{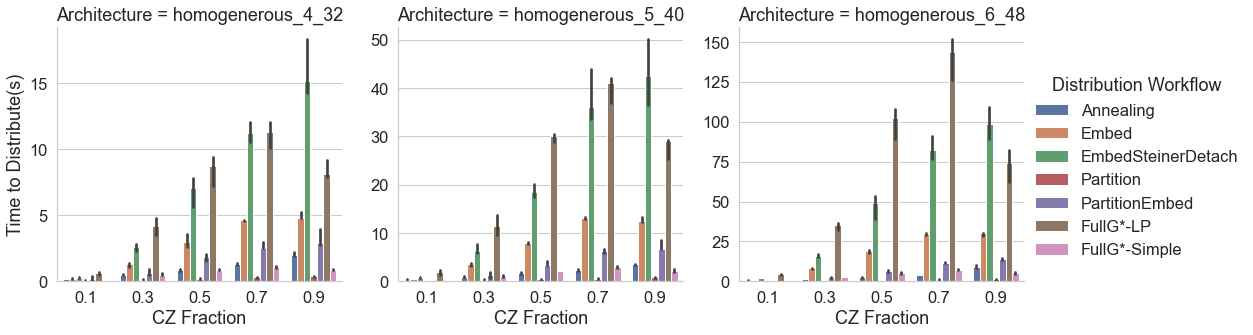}
         \caption{Time to generate distribution, measured in seconds.}
         \label{fig:benchmarks homogeneous cz fraction time}
    \end{subfigure}
    \caption{\textbf{Distribution techniques applied to homogeneous networks and CZ fraction circuits.} Here we use the notation where \texttt{homogeneous\_n\_m} is a homogeneous network connecting \texttt{n} modules in a network with a total of \texttt{m} qubits. Bars indicate the median over 5 circuits. Error bars indicate 75\% percentile range.}
    \label{fig:benchmarks homogeneous cz fraction}
\end{figure*}

\anneal{} performs the worst overall, which is to be expected as the methods used are particularly general. However, \anneal{} is particularly sensitive to the values of hyper-parameters, particularly the number of annealing iterations performed. Hence, these results may be improved by increasing the number of iterations. Here, the number of iterations is chosen so that the time taken by \anneal{} is roughly comparable to those of the best performing unrefined distribution workflows, as seen in \cref{fig:benchmarks homogeneous cz fraction time}. \hp{} performs the quickest across circuit sizes and \CZ{} fractions, while the scaling of \gslp{} and \cesd{} is the worst. However, as no workflow takes more than a few minutes to complete, the time taken is acceptable in all cases.

\ce{} performs poorly in the results of \cref{fig:benchmarks homogeneous cz fraction cost}. This is unsurprising as it corresponds to the original work from~\cite{Junyi2022} which was designed to work best with 2 modules, where detached gates need not be considered. However \cesd{} significantly improves upon \ce{}, demonstrating the significant potential gains to be made from the use to detached gates. Indeed, in the case of 2 modules, as seen in \cref{fig:benchmarks homogeneous cz fraction 2 modules}, \ce{} performs the best (particularly in the regime of 50 qubits and \CZ{} fraction of 0.5 and 0.7). In this case \cesd{} does not improve the results, as is to be expected since in the 2 module case there is no opportunity for detached gates.

\begin{figure*}
    \centering
    \begin{subfigure}[b]{\textwidth}
         \centering
         \includegraphics[width=0.7\textwidth]{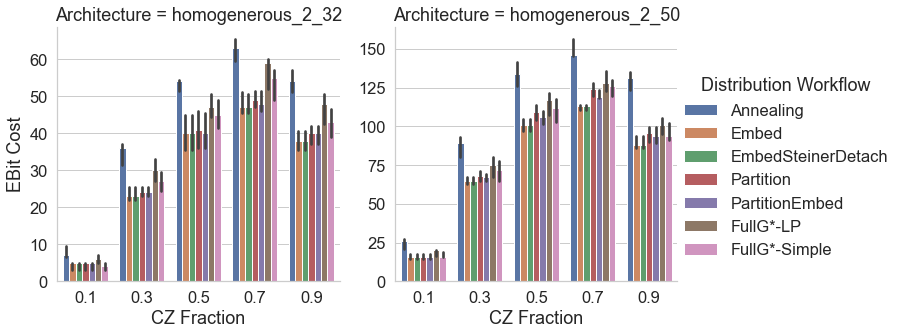}
         \caption{Ebit cost.}
         \label{fig:benchmarks homogeneous cz fraction cost 2 modules}
     \end{subfigure}
     
     \vspace{10pt}
     
     \begin{subfigure}[b]{\textwidth}
         \centering
         \includegraphics[width=0.7\textwidth]{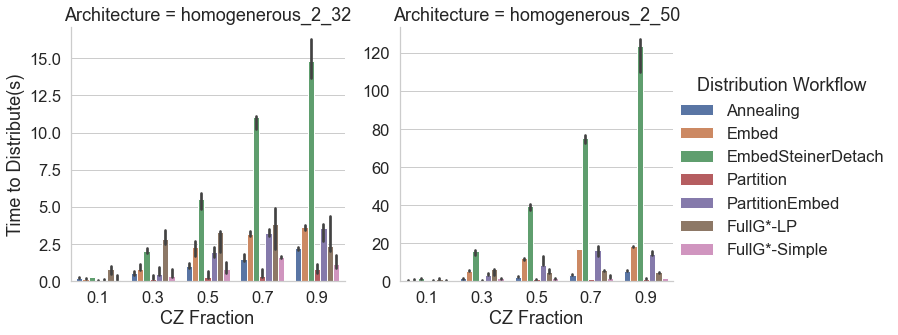}
         \caption{Time to generate distribution, measured in seconds.}
         \label{fig:benchmarks homogeneous cz fraction time 2 modules}
    \end{subfigure}
    \caption{\textbf{Distribution techniques applied to homogeneous networks and CZ fraction circuits over 2 modules.} Here we use the notation that \texttt{homogeneous\_n\_m} is a homogeneous network connecting \texttt{n} nodes in a network with a total of \texttt{m} qubits. Bars indicate the median over 5 circuits. Error bars indicate 75\% percentile range.}
    \label{fig:benchmarks homogeneous cz fraction 2 modules}
\end{figure*}

In the case of networks containing more than 2 modules, \pe{} barely improves upon \hp{}. This may be because \hp{} produces many detached gates which cannot be embedded by the embedding refinement pass. In the case of 2 server networks, where no gates are detached, \pe{} mildly improves upon \hp{}, but does not outperform \ce{}. This demonstrates that embedding can be beneficial when sequences of gates act between 2 modules, but implies that embedding should be considered in the first instance on such networks, rather than through refinement.

We consider the performance of these techniques on the Quantum Volume and Pauli Gadget circuit classes, giving the results in \cref{fig:benchmarks homogeneous random}. Here we consider only network with greater than 2 modules, and so do not consider \ce{} which performs well only on 2 module networks. As these circuits have a significantly larger number of gates than the \CZ{} Fraction circuits we consider only the quicker distribution workflows, namely \gss{}, \hp{}, \pe{}, and \anneal{}.

\begin{figure*}
    \centering
    \begin{subfigure}[b]{\textwidth}
         \centering
         \includegraphics[width=0.7\textwidth]{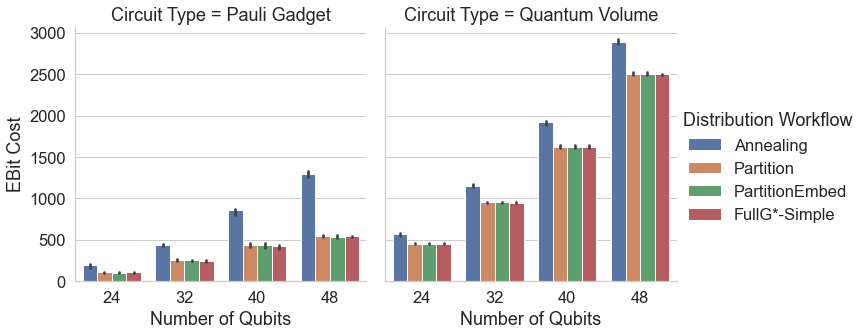}
         \caption{Ebit cost.}
         \label{fig:benchmarks homogeneous random cost}
     \end{subfigure}
     
     \vspace{5pt}
     
     \begin{subfigure}[b]{\textwidth}
         \centering
         \includegraphics[width=0.7\textwidth]{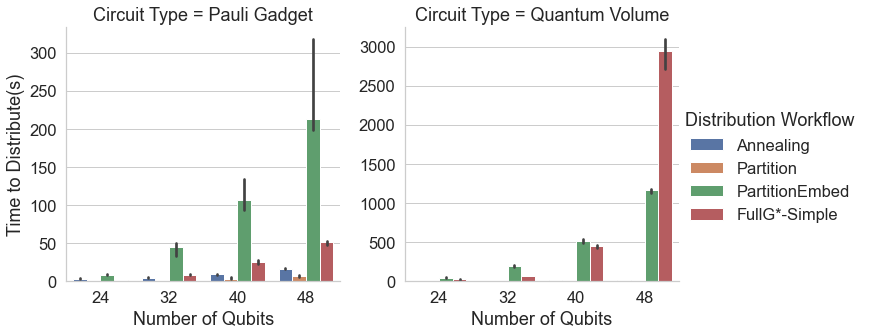}
         \caption{Time to generate distribution, measured in seconds.}
         \label{fig:benchmarks homogeneous random time}
    \end{subfigure}
    \caption{\textbf{Distribution techniques applied to homogeneous networks and Quantum Volume and Pauli Gadget circuits.} Here we use homogeneous networks built of 3, 4, 5 and 6 modules, each with 8 qubits. Each sample in the experiment corresponds to a single circuit, with 5 samples per bar. Bars indicate the median over five circuits. Error bars indicate 75\% percentile range.}
    \label{fig:benchmarks homogeneous random}
\end{figure*}

Note that the cost of distributing Pauli Gadget circuits is cheaper for a similar total number of two-qubit gates than the cost of distributing Quantum Volume circuits. Refer to \cref{fig:benchmarks circuits size} for details on comparative 2-qubit gate counts. This is to be expected since the structure of Pauli Gadget circuits, having long sequences of \CZ{} gates, allows for the construction of larger distributable packets. For the same reason, Pauli Gadget circuits may be distributed more quickly.

In \cref{fig:benchmarks homogeneous random} we see a similar pattern to the relative performance of the schemes as we saw in \cref{fig:benchmarks homogeneous cz fraction}, namely that there is no significant difference in the e-bit costs of the distributions produced by each workflow, apart from that \anneal{} has a higher cost. \hp{} performs best if both the ebit cost and time taken are considered. \gss{} performs similarly well as measured by ebit cost, but the time required to distribute with \gss{} scales worse as the number of distributable packets becomes very large, as is the case for the larger Quantum Volume circuits.

\subsubsection{Heterogeneous networks}
\label{sec:benchmarks results heterogeneous}

Here we compare the performance of \ce{}, \ces{}, \cesd{} \hp{}, \ph{}, \pe{}, and \phe{}, each of which is capable of performing circuit distribution over heterogeneous networks (although \ce{} and \hp{} are not designed for them). We do not include results for \anneal{} in the plots of this section, as in each case it is outperformed by \ph{}. We use networks with 3, 4, and 5 modules, each with an average of 6 computational qubits per module. Here we do not bound the size of the link qubit register, instead exploring these bounds in \cref{sec:limited links}. For each network size we generate 5 random instances of each of the heterogeneous networks described in \cref{sec:benchmarks networks}. The results of these benchmarks can be found in \cref{fig:benchmarks heterogenous random,fig:benchmarks heterogeneous cz fraction}, and our findings are detailed below.

\begin{figure*}
     \centering
     \begin{subfigure}[b]{\textwidth}
         \centering
         \includegraphics[width=\textwidth]{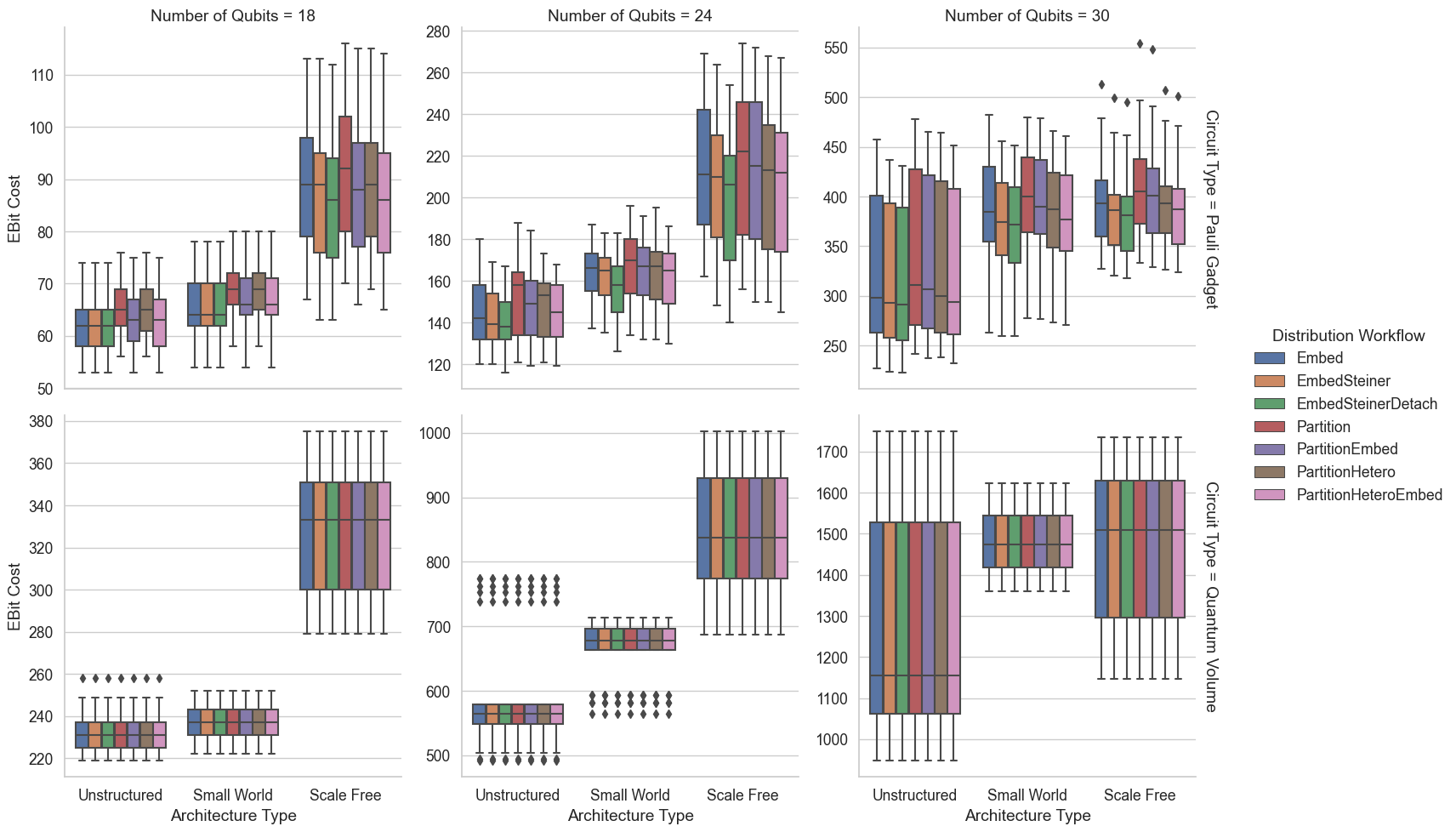}
         \caption{Ebit cost. Boxes give median and interquartile range. Whiskers extend to the last data point within 2.5 times the interquartile range from the median.}
        \label{fig:benchmarks heterogenous random cost}
     \end{subfigure}

     \vspace{5pt}
     
     \begin{subfigure}[b]{\textwidth}
         \centering
         \includegraphics[width=\textwidth]{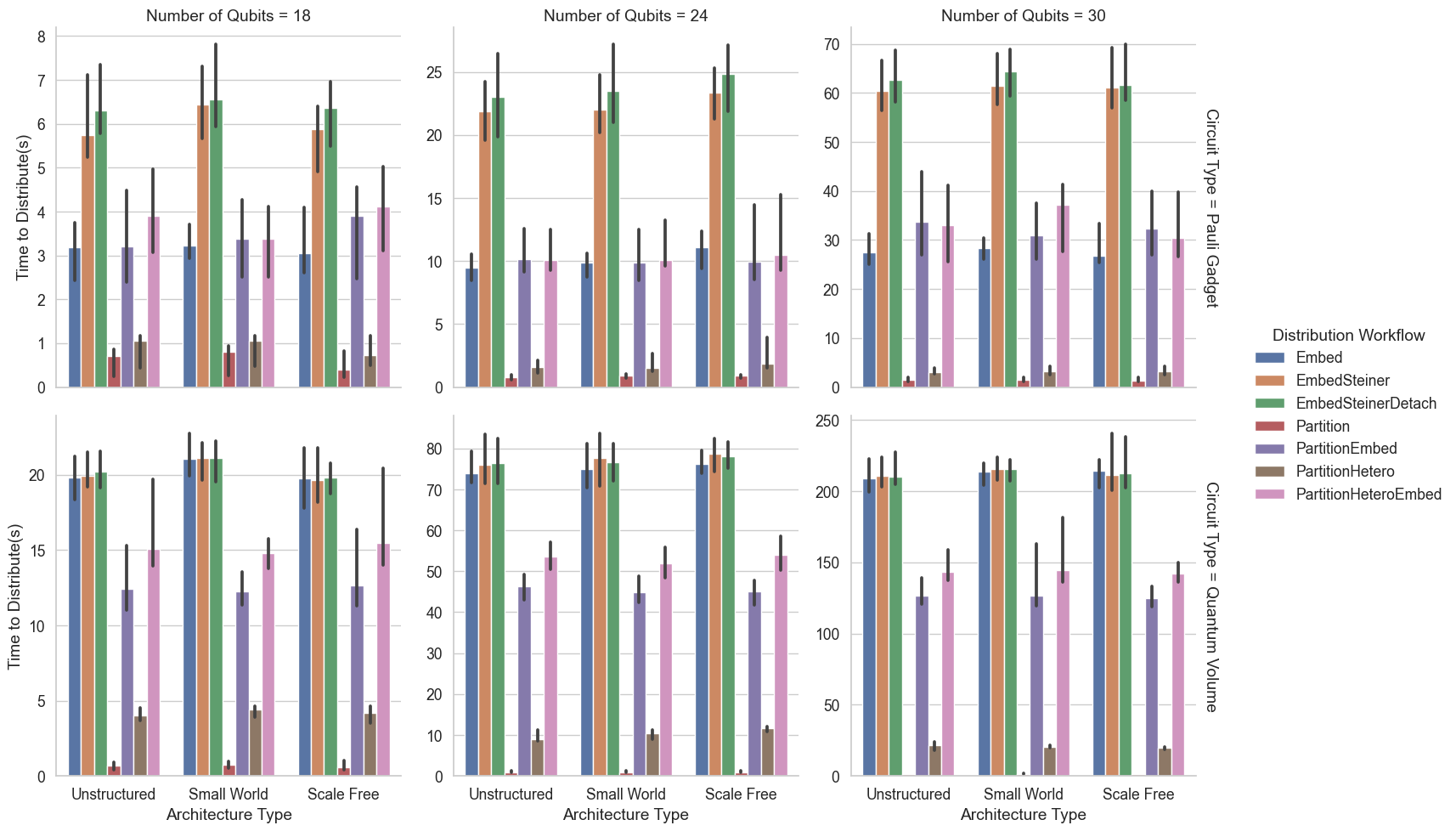}
         \caption{Time to generate distribution, measured in seconds. Bars indicate the median over 5 circuits. Error bars indicate 75\% percentile range.}
         \label{fig:benchmarks heterogenous random time}
     \end{subfigure}
     \caption{\textbf{Distribution over heterogeneous networks.} Here we use heterogeneous networks built of 3, 4, and 5 modules, each with an average of 6 qubits. Each sample in the experiment corresponds to a single circuit-network pair. Each bar/box considers 5 circuits and 5 networks, giving a total of 25 circuit-network pairs per bar/box.}
    \label{fig:benchmarks heterogenous random}
\end{figure*}

\begin{figure*}
    \centering
     \begin{subfigure}[b]{\textwidth}
         \centering
         \includegraphics[width=\textwidth]{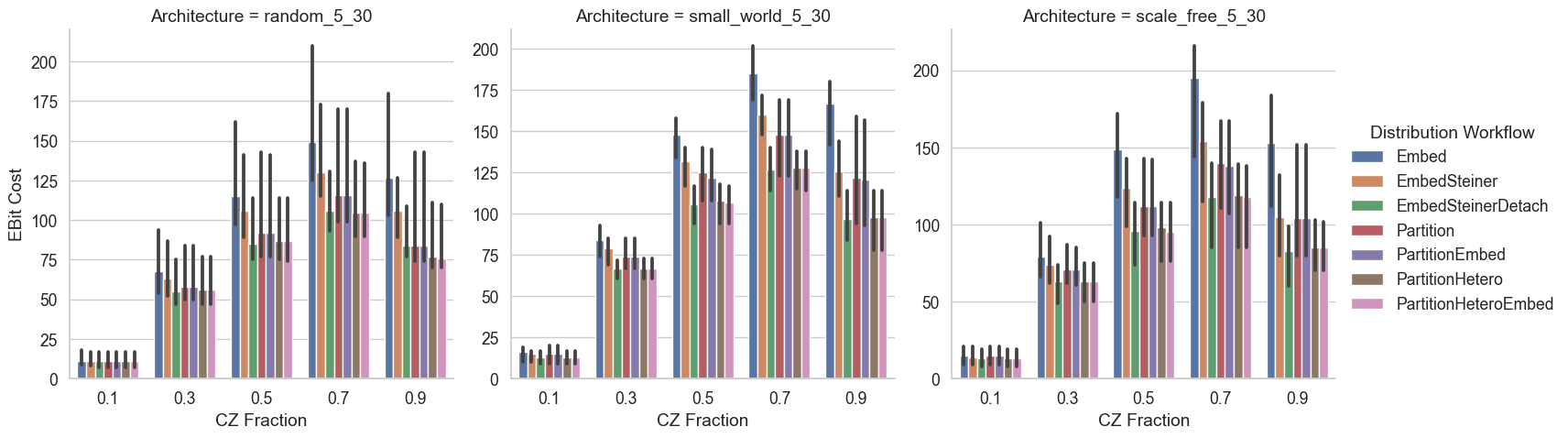}
         \caption{Ebit cost.}
         \label{fig:benchmarks heterogeneous cz fraction cost}
     \end{subfigure}
     
     \vspace{10pt}
     
     \begin{subfigure}[b]{\textwidth}
         \centering
         \includegraphics[width=\textwidth]{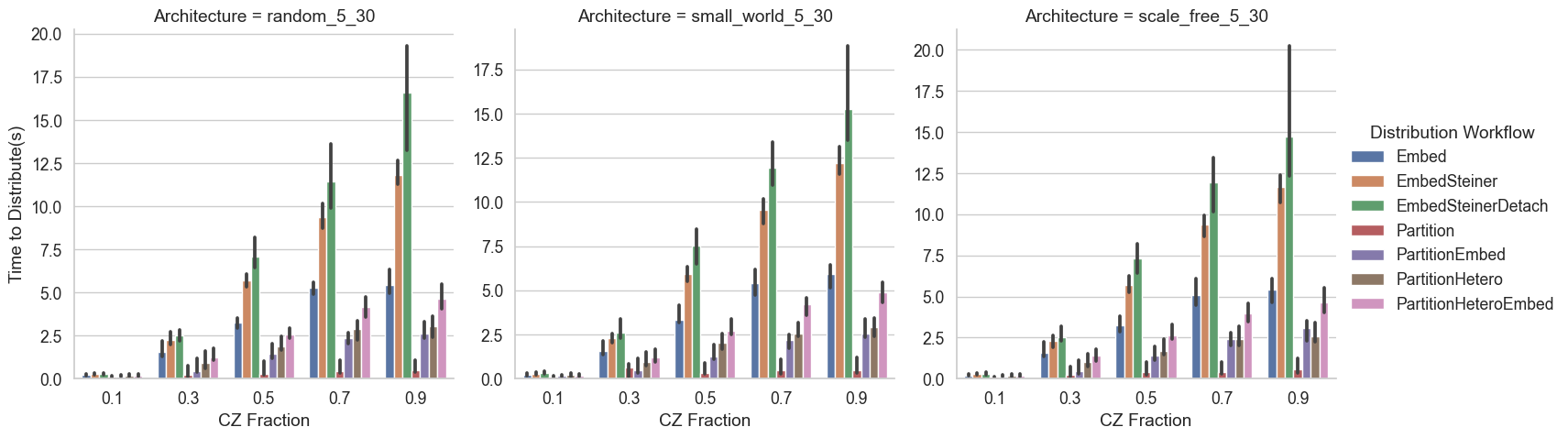}
         \caption{Time to generate distribution, measured in seconds.}
         \label{fig:benchmarks heterogeneous cz fraction time}
    \end{subfigure}
    \caption{\textbf{Distribution techniques applied to heterogeneous networks and CZ fraction circuits.} Here we use the notation where \texttt{type\_n\_m} is a network of type \texttt{type} connecting \texttt{n} modules in a network with a total of \texttt{m} qubits. Bars indicate the median over 5 circuits. Error bars indicate 75\% percentile range.}
    \label{fig:benchmarks heterogeneous cz fraction}
\end{figure*}

\gss{} and \gslp{} are not suited to heterogeneous networks, and so the most relevant comparable result is that of \cite{Sundaram2022GeneralDistribution}. Unfortunately, we were not able to access the implementation of the work of \cite{Sundaram2022GeneralDistribution} for comparison. The latter work does not make use of Steiner trees for entanglement distribution, which are utilised by all the the schemes presented in this section except \ce{}. The work of \cite{Sundaram2022GeneralDistribution} does not consider embedding either, which is considered by \ce{}, \ces{}, \cesd{}, \pe{}, and \phe{}, and is shown to provide a reduction in ebit cost. As such we expect our techniques to compare favourably to those of \cite{Sundaram2022GeneralDistribution}.

\paragraph{Refinement has little effect on Quantum Volume circuits.} We expect that distributable packets are unavoidably small in the case of Quantum Volume circuits since there are few consecutive \CU{} gates in the circuits and few valid embedding units: the phases of \RZ{} gates will rarely satisfy condition (d) from Lemma~\ref{lem:embedding_cond}. In \cref{fig:benchmarks heterogenous random cost} this manifests in there being no gain from using refinement passes targeted at the use of Steiner trees and embedding. 

Additionally, no benefit is found in these circuits when performing boundary reallocation targeted at optimising for the network topology (\ph{}) and detached gates (\cesd{}). This again reflects that the hyperedges are too small (often just edges from gate-vertex to qubit-vertex) which, combined with the uniformly random connectivity of the circuit, leads to no window for improvement of the vertex allocation.

\paragraph{Each refinement improves the median cost of Pauli Gadget circuits.} As opposed to Quantum Volume circuits, distributable packets in Pauli Gadget circuits are relatively large, and can be beneficially combined. This is shown in the improvement achieved in \cref{fig:benchmarks heterogenous random cost,fig:benchmarks heterogeneous cz fraction cost} by employing refinement passes making use of Steiner trees, detached gates and embedding. 

\paragraph{Pauli Gadget circuits are cheaper and quicker to distribute.} \cref{fig:benchmarks heterogenous random cost} demonstrates that, as a result of Pauli Gadget circuits having larger distributable packets, the cost of distribution of Pauli Gadget circuits is much less than that of Quantum Volume circuits of similar size. Likewise, as seen in \cref{fig:benchmarks heterogenous random time}, the time required to distribute Pauli gadget circuits is shorter since run time scales primarily with respect to the number of packets, rather than the number of qubits or gates in the circuit.

\paragraph{CZ Fraction circuits on networks with more than 2 modules do not benefit greatly form embedding.} As observed initially in \cref{fig:benchmarks homogeneous cz fraction cost}, we see again in \cref{fig:benchmarks heterogeneous cz fraction cost} that refinement to make use of embedding has little impact on the resulting cost of distributing CZ fraction circuits onto networks with more than 2 modules. This identifies a middle ground between the more structured Pauli Gadget circuits, which do benefit from embedding, and the larger gate set of the Quantum Volume circuits, which do not benefit from refinement of any kind.

\paragraph{Techniques combined perform best} We see that \cesd{} typically perform as well or better than the other workflows. This demonstrates the benefit of combining the use of detached gates, Steiner trees, and embedding, and that no one or two alone would perform best. That \cesd{} mildly outperforms \phe{} on average --- which also makes use of detached gates, Steiner trees and embedding ---  in the Pauli Gadget results of \cref{fig:benchmarks heterogenous random cost} indicates that embedding is hard to capture in a refinement pass, so it should instead be optimised for in the first instance. 

\subsubsection{Chemically-Aware Ansatz}
\label{sec:benchmarks results chemistry}

We explore the performance of our approaches in the particular case of a chemically-aware unitary coupled cluster singles and doubles ansatz \cite{kham2022}. We use the example of the minimal basis $\text{H}_2 \text{O}$ molecule with $\text{C}_{2v}$ point group symmetry and the 6 electrons in 5 spatial orbital (6e, 5o) active space. The corresponding circuit contains 10 qubits, and is built from Pauli gadgets selected to reflect the symmetries of the system. In the gateset \gateset{} the circuit contains 463 2-qubit gates. 

We distribute this circuit onto the networks of 11 qubits depicted in \cref{fig:benchmarks chimistry aware}, without bounds on the link qubit register sizes. The results are listed in \cref{tab:chem aware}. In the results of \cref{sec:benchmarks results heterogeneous}, \cref{sec:benchmarks results homogeneous} and \cref{sec:limited links} the number of qubits in the circuit matches the total number of computation qubits in the network. However our tools are capable of managing situations where there are more computational qubits in the network than are required by the circuit, as demonstrated here.

As expected and indicated by the results of \cref{sec:benchmarks results heterogeneous}, we see that the ebit cost decreases with additional refinement. Here it is noticeable that embedding is beneficial, both when introduced as part of a refinement pass, and when introduced during an initial circuit distribution. This shows that real application have circuit structures which benefit from embedding. Indeed it is the case that \cesd{}, which introduces embedding in the first instance, performs best.




\begin{figure}
    \centering
    \begin{subfigure}[b]{0.45\columnwidth}
        \begin{tikzpicture}[thick, scale=1.5]
            \coordinate (a) at (0.587,-0.809);
            \coordinate (b) at (0.951,0.309);
            \coordinate (c) at (0,1);
            \coordinate (d) at (-0.951,0.309);
            \coordinate (e) at (-0.587,-0.809);
            
            \draw (a) -- (d);

            \draw (b) -- (c);

            \draw (c) -- (d);

            \draw (d) -- (e);

            \filldraw[black, fill=white] (a) circle (4pt) node {3};
            \filldraw[black, fill=white] (b) circle (4pt) node {2};
            \filldraw[black, fill=white] (c) circle (4pt) node {2};
            \filldraw[black, fill=white] (d) circle (4pt) node {2};
            \filldraw[black, fill=white] (e) circle (4pt) node {2};
        \end{tikzpicture}
        \caption{Network 1}
    \end{subfigure}
    \hfill
    \begin{subfigure}[b]{0.45\columnwidth}
        \begin{tikzpicture}[thick, scale=1.5]
            \coordinate (a) at (0.587,-0.809);
            \coordinate (b) at (0.951,0.309);
            \coordinate (c) at (0,1);
            \coordinate (d) at (-0.951,0.309);
            \coordinate (e) at (-0.587,-0.809);
            
            \draw (a) -- (b);
            \draw (a) -- (d);
            \draw (a) -- (e);


            \draw (c) -- (d);

            \draw (d) -- (e);

            \filldraw[black, fill=white] (a) circle (4pt) node {3};
            \filldraw[black, fill=white] (b) circle (4pt) node {2};
            \filldraw[black, fill=white] (c) circle (4pt) node {2};
            \filldraw[black, fill=white] (d) circle (4pt) node {2};
            \filldraw[black, fill=white] (e) circle (4pt) node {2};
        \end{tikzpicture}
        \caption{Network 2}
    \end{subfigure}
    \caption{\textbf{Networks for chemistry aware experiments.} Numbers in vertices indicate the number of qubits in each module; edges indicate connections alone which ebits can be established.}
    \label{fig:benchmarks chimistry aware}
\end{figure}
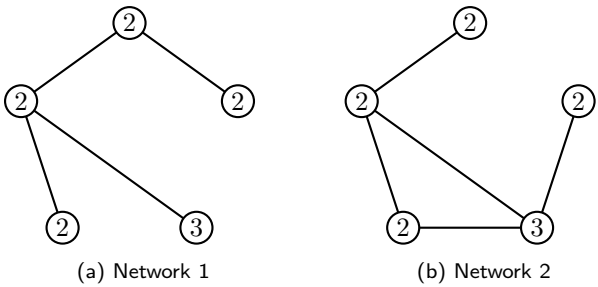


\begin{table*}
    \begin{tabularx}{\hsize}{bssssssss}
        \toprule
        Workflow & \multicolumn{2}{m}{Ebits} & \multicolumn{2}{m}{Detached}  & \multicolumn{2}{m}{Non-Local} & \multicolumn{2}{m}{Hyperedges} \\
         & 1 & 2 & 1 & 2 & 1 & 2 & 1 & 2 \\
        \midrule
        \ce{} & 238 & 172 & 0 & 0 & 343 & 340 & 631 & 632 \\
        \ces{} & 233 & 164 & 0 & 0 & 343 & 340 & 355 & 355 \\
        \cesd{} & 230 & 160 & 13 & 19 & 344 & 342 & 355 & 355 \\
        \midrule
        \hp{} & 295 & 196 & 28 & 28 & 342 & 342 & 397 & 397 \\
        \pe{} & 276 & 183 & 28 & 28 & 342 & 342 & 352 & 352 \\
        \ph{} & 256 & 189 & 35 & 27 & 350 & 344 & 397 & 397 \\
        \phe{} & 253 & 182 & 35 & 27 & 350 & 344 & 353 & 357 \\
        \bottomrule
    \end{tabularx}
    \caption{\textbf{Distributing chemistry aware ansatz circuits.} Networks 1 and 2 are found in \cref{fig:benchmarks chimistry aware}.}
    \label{tab:chem aware}
\end{table*}

\section{Conclusion and future work}
\label{sec:conclusion}

In this work we consider the distribution of quantum circuits over heterogeneous networks. We propose a collection of methods for distributing a given quantum circuit over an arbitrary network in a way  which minimises the number of ebits required. We make these methods available through \pytketdqc{}. 

Our first contribution is to introduce two workflows, \anneal{} and \hp{}, which perform quantum circuit distribution over heterogeneous networks in a way which makes use of detached gates. Secondly, where previous work had made use of detached gates or embedding, we present approaches to combining both. We do so by starting from distribution workflows this make use of either and applying rounds of refinement to make the most use of the other. Finally by proposing and incorporating entanglement distribution via Steiner trees, and by developing methods to combine their use with embedding, we further improve our solutions.

We extensively benchmark our distribution workflows on a selection of random and application motivated circuits. We identify that the best workflow to utilise on bipartite networks is \ce{}, while for larger homogeneous networks \hp{} is best. For structured application motivated circuits on heterogeneous networks \cesd{} is best, while for unstructured Quantum Volume circuits it is best to simply use the fastest workflow, in this case \hp{}.

In the future, optimisation strategies that can take into account the bound to the link qubit registers should be explored further. We are aware of two papers that do so, namely~\cite{Sundaram2022GeneralDistribution, Junyi2022}; however, the approach from~\cite{Sundaram2022GeneralDistribution} does not consider the embedding technique nor Steiner trees, while~\cite{Junyi2022} targets networks with only two modules. Moreover, even though the approach from \cite{Sundaram2022GeneralDistribution} tends to yield solutions that meet the specified bound to the link qubit register, this is not guaranteed --- in certain cases, it is necessary to split some of the distributable packets in a similar way we discuss in \cref{sec:limited links}. 

Future work may also consider preprocessing of the circuit to facilitate less costly distributions. This is particularly applicable to Pauli Gadgets, which may be decomposed in a variety of ways \cite{Cowtan_2020}, each of which may be more or less suited to distribution. Finally, we encourage the investigation of dynamical quantum circuit distribution, which combines gate teleportation and qubit teleportation. In \cite{Sundaram2022GeneralDistribution, Andr_s_Mart_nez_2019} the authors propose approaches to doing so, suggesting the static distribution of segments of circuits, stitched together via qubit teleportation. The work of this paper can be straightforwardly used as a static distributor in this framework, obtaining similar gains as those reported in~\cite{Sundaram2022GeneralDistribution, Andr_s_Mart_nez_2019}. We expect that approaches capable of freely interleaving qubit teleportation and EJPP processes be even more beneficial, and we suggest this be the most pressing line of further work.

\paragraph{Code Availability}

The techniques outlined in \cref{sec:solutions} are implemented in \pytketdqc{}, which can be found in \repo{} along with example notebooks. Documentation for \pytketdqc{} can found at \url{https://cqcl.github.io/pytket-dqc/}. The results of the benchmarks in \cref{sec:benchmarks} can be found in \url{https://github.com/CQCL/pytket-dqc_experiment_data}.

\paragraph{Benchmark Tools}

The results in \cref{sec:benchmarks results} were obtained using a MacBook Pro with a 2.3 GHz Dual-Core Intel Core i5 processor and 8 GB 2133 MHz LPDDR3 memory. Time to generate distribution in each plot refers to the time taken by this machine.

\paragraph{Acknowledgements}

The authors thank Ranjani Sundaram,  Himanshu Gupta,  and C.R.Ramakrishnan for their insights on their own work, and for sharing the related code. We also acknowledge the contributions made by Kosuke Matsui and Akihito Soeda during early conversations about the project. Thanks to Matty Hoban and Yao Tang for their careful proofreading of this manuscript.

JYW is supported by Ministry of Science and Technology, Taiwan, R.O.C. under Grant no. NSTC 110-2112-M-032-005-MY3, 111-2923-M-032-002-MY5 and 111-2119-M-008-002. TF and MM are supported by: MEXT Quantum Leap Flag-ship Program (MEXT QLEAP) JPMXS0118069605, JPMXS0120351339, Japan Society for the Promotion of Science 21H03394. 

\printbibliography

\appendix
\addtocontents{toc}{\protect\setcounter{tocdepth}{0}}

\section{Refinement}
\label{sec:refinement}

Our software, \pytketdqc{}, contains several \texttt{Refiners} which act on an already valid distribution of a circuit and further reduce its ebit cost. For instance, the approach described in \cref{sec:boundary_realloc} is implemented as a refiner. Refiners that have not been explained in the main text are briefly described here; some of these are used in our default workflows listed in \cref{sec:benchmarks workflows}. Note that sequencing and repeating these refiners may result in greater improvement than a single application, and there exist functionality in \pytketdqc{} for constructing such sequences.

\subsection{Detached gate identification}
\label{sec:refinement detatched}

The approach described in \cref{sec:boundary_realloc} refines a hypergraph's partition taking into account the heterogeneous network. It can be repurposed for the task of identifying opportunities where non-local gates may be implemented in a detached manner, \ie{} as in \cref{fig:detached_gate}. For this purpose, we impose that qubit-vertices of the hypergraph cannot have their allocation changed, so that the set of non-local gates remains the same. Furthermore, we impose that gate-vertices corresponding to embedded gates are not reallocated either, since Algorithm~\ref{alg:ALAP} is not capable of accurately estimating the cost of embedding a detached gate. As such, we can apply this approach as a refiner at the end of any workflow, reallocating gate-vertices that do not have a risk of detrimentally interfering with previous optimisations. 

This refiner is meant to be applied at the end of any workflow that employs the vertex covering approach discussed in \cref{sec:vertex_cover_embedding}. On its own, the vertex covering approach cannot take advantage of distribution via detached gates \cref{fig:detached_gate}, but such opportunities can be easily identified by hypergraph partitioning approaches, since it is just a matter of allocating the corresponding gate-vertex to a module other than where its qubits are assigned to. Since the approach from \cref{sec:boundary_realloc} only changes the allocation of a vertex when doing so reduces the ebit cost of the distribution (calculated using the approach from \cref{sec:ALAP}) and non-local gates that may be implemented in a detached manner are necessarily in the boundary of the partition, the refiner will be able to identify opportunities for these and reduce the cost of the distribution accordingly.

\subsection{Eager H-type merging}
\label{sec:refinement embedding}
Eager H-type merging refers to the merging of distributable packets via embedding. The refiner scans the circuit qubit by qubit, packet by packet, from start to end, finding opportunities where embedding can be used to merge distributable packets of the given solution.
For a given packet $P_0$ we first identify the next packet $P_1$ that can be merged via embedding (if any). We check whether an embedding conflict would be created by said merging and whether the embedding unit includes any detached gates.
If neither, the refiner merges $P_0$ and $P_1$ and, otherwise the packets are not merged. Regardless of the outcome, the refiner continues the search until all pairs of packets have been considered.

This refiner allows for the use of the embedding technique after workflows that do not use/optimise for it. Since it does not alter the way each of the non-local gates are distributed --- it only extends the lifespan of link qubits --- it can easily be applied at the end of any workflow. This comes at the disadvantage of not exploiting the potential of the embedding technique to the fullest. If embedding is expected to be the main source of ebit cost reduction on a given distribution, a workflow such as \ce{} or any of its derivatives would be preferable.

\subsection{D-type merging}
\label{sec:refinement steiner}

D-type merging refers to the merging of hyperedges when doing so does not require additional embedding. In particular two hyperedges on the same qubit can be D-type merged when two \CU{} gates, one from each packet, act consecutively on said qubit with no \H{} gates acting between them.

D-type merging has the effect of merging multiple distributable packets --- that may share their qubits with different modules --- into a single hyperedge. This has the advantage of allowing for greater opportunity to reduce ebit cost through the use of gate distribution via Steiner trees, as discussed in \cref{sec:Steiner}. As such, we recommend the use of D-type merging on workflows employing the vertex covering approach of \cref{sec:vertex_cover} which, on its own, would produce hyperedges involving only two modules each, preventing the use of optimisations based on Steiner trees.

There are two D-type merging refiners in \pytketdqc{}: \texttt{NeighbouringDTypeMerge} and \texttt{IntertwinedDTypeMerge}, which differ only in the relative positioning of the distributable packets which they merge. Each refiner iterates through the packets acting on a qubit, merging them when a D-type merge is possible.

\section{Building the distributed circuit}
\label{sec:to_pytket_circuit}

The outcome of each of the approaches discussed in this paper is a \texttt{Distribution} which, as established in Definition~\ref{def:distribution}, corresponds to a hypergraph along with an allocation of its vertices to modules. However, we ultimately want to convert this abstract data structure to an actual quantum circuit; a method to do so is provided within \pytketdqc{}.
Algorithm~\ref{alg:ALAP} described how, given a hyperedge and its allocation of vertices to modules, we can distribute its corresponding subcircuit, implementing the non-local gates corresponding to its gate-vertices via EJPP protocols and embedding the rest as appropriate. 
In order to distribute the whole circuit, we apply Algorithm~\ref{alg:ALAP} on the input circuit once per hyperedge in the \texttt{Distribution}. There are some subtleties that were omitted in the main text for the sake of brevity and are detailed below.

\paragraph{Correction gates.} These are extra gates that must be applied along with embedding units in order to preserve circuit equality. Recall that the input circuit has been rebased to the \gateset{} gateset and that, if the embedding unit commutes with the starting process, no correction gates are required (this follows immediately from Definition~\ref{def:embedding_unit}). Thus, correction gates are only required by embedding units that begin and end with an \H{} gate. \cref{fig:embedding_example} provides the two most simple examples where correction gates appear; from these, the general case can be inferred. 

Due to conditions (c) and (d) of Lemma~\ref{lem:embedding_cond} we only need to concern ourselves with correcting \H{}, \Z{} and \CZ{} gates. In particular, whenever an \H{} (or \Z{}) gate acting on $\hat{q}$ is being embedded, we must apply an \H{} (respectively, \Z{}) gate on each link qubit that is currently entangled with $\hat{q}$. In the case of a \CZ{} gate, let $\hat{q}$ be the qubit that is being shared and let $q'$ be the other qubit the \CZ{} gate acts on; we require one correction \CZ{} gate per link qubit currently entangled with $\hat{q}$, with the gate acting on such a link qubit and $q'$. As shown in \cref{fig:embedding_example}, the correction \CZ{} gates are local; this is guaranteed by condition (b) of Lemma~\ref{lem:embedding_cond}.
Repeating this process above for every gate within an embedding unit provides all of the correction gates that are required by Algorithm~\ref{alg:ALAP}.

\paragraph{Ending processes.} When using Algorithm~\ref{alg:ALAP}, a starting process may first entangle a proxy link qubit $q$ with another link qubit $q'$, and then have the ending process disentangling $q$ appear before the disentanglement of $q'$. In such a situation, $q'$ has lost its immediate predecessor in the entanglement chain, and it may be unclear how to disentangle it. Fortunately, the only gate to be applied on $q$ during the ending process of $q'$ is a classically controlled $Z$ gate (see \cref{fig:ejpp}) which may be equivalently applied on the root of the Steiner tree: the circuit qubit $\hat{q}$. Since ending processes only use LOCC, the network architecture does not pose an obstacle to their implementation and, hence, there is no dependency between ending processes. 

\paragraph{Intertwined embeddings.} In situations such as the one depicted in \cref{fig:intertwined} where two distributable packets are rooted on the same qubit, applying Algorithm~\ref{alg:ALAP} on one of them yields a circuit (depicted in \cref{fig:intertwined}b) that would require embedding a starting process for the second packet to be distributed. The following equality can be derived by manipulating the circuit required to implement a starting process (see \cref{fig:ejpp}).

{\centering
\begin{tikzpicture}
    \node (squiggle) {
        \begin{quantikz}[column sep=4mm]
            \qw & \qw \arrow[d, squiggly] & \qw \\[20pt]
             & & \qw
        \end{quantikz}
    };
    \node[right=2mm of squiggle] (eq1) {\Large $=$};
    \node[right=2mm of eq1] (ebit) {
        \begin{quantikz}[column sep=2mm]
            \qw & \qw & \ctrl{1} & \qw & \qw \\
            & \dEbit{} \arrow[d, dash, squiggly] & \targ{} & \meter[style={scale=0.75, thin}]{} \arrow[d, dash, thick, xshift=1.25pt] \arrow[d, dash, thick, xshift=-1.25pt] & \\
            & \dEbit{} & \qw & \dGate{X} & \qw
        \end{quantikz}
    };
    \node[right=2mm of ebit] (eq2) {\Large $=$};
    \node[right=2mm of eq2] {
        \begin{quantikz}[column sep=4mm]
            \qw & \qw & \ctrl{1} & \qw \\[20pt]
             & \lstick{\ket{0}} & \targ{} & \qw
        \end{quantikz}
    };
\end{tikzpicture}
}

This means that --- for matters of embedding --- we can treat starting processes just as if they were non-local \CX{} gates. Since a \CX{} gate is equivalent to a \CZ{} sandwiched by \H{} gates, the approach discussed in the main text can be applied directly to embedding units containing starting processes. Doing so yields the distributed circuit from \cref{fig:intertwined}c. Proving that the resulting \CX{} correction gate is always local is nontrivial, but it follows from the fact that this intertwining of embeddings can only occur if there is a \CZ{} gate such as $y$ in \cref{fig:intertwined}c that is distributable in one packet and embedded in the other. Then, due to condition (b) of Lemma~\ref{lem:distributable_cond} and condition (b) of Lemma~\ref{lem:embedding_cond}, the remote module \texttt{B} that the qubit is being shared with must be the same for both packets. Consequently, both link qubits live in the same module \texttt{B} and the correcting \CX{} gate is local in \texttt{B}.
In the case of ending processes, a similar argument holds, although a simpler approach is to realise that the gate that would need to be embedded is just a classically controlled \Z{} gate and, consequently, its correction is straightforward.

\begin{figure*}
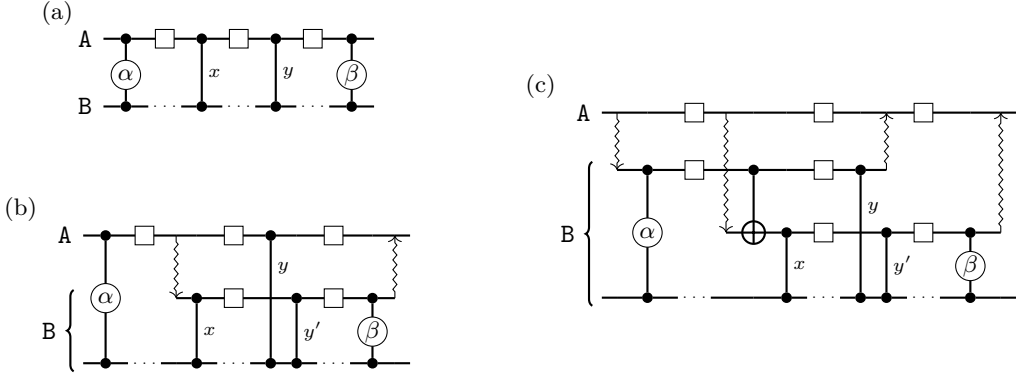

    \centering    
    \begin{tikzpicture}
    \node (input) {\input{images/intertwined/input}};
    \node[below=10mm of input] (intermediate) {\input{images/intertwined/intermediate}};
    \node[above right=-15mm and 15mm of intermediate] (final) {\input{images/intertwined/final}};
    \node[above left=-3mm and -3mm of input] {\small (a)};
    \node[above left=-3mm and -3mm of intermediate] {\small (b)};
    \node[above left=-3mm and -3mm of final] {\small (c)};
\end{tikzpicture}
    \caption{\textbf{Distribution with intertwined embeddings.} (a) Input circuit; two distributable packets are considered, both rooted on the qubit in $\texttt{A}$: $P_0 = \{\alpha, y\}$ and $P_1 = \{x, \beta\}$. (b) Circuit after distributing $P_1$; gate $y$ is embedded. (c) Circuit after distributing $P_0$ as well; we need to embed the starting process of $P_1$, which requires a local $\CX$ correction gate.}
    \label{fig:intertwined}
\end{figure*}

\section{Limited Link Qubits}
\label{sec:limited links}

In \cref{sec:benchmarks} no bound on the number of available link qubits was imposed. In practice two considerations bound this quantity:
\begin{itemize}
    \item Each module has a fixed total number of qubits, and a register of unlimited size dedicated to link qubits would be infeasible. The sum of the number of computation qubits and link qubits required by the distributed circuit should be less than the total number of qubits in each module.
    \item The sum of the number of computation qubits and the number of link qubits in the largest module should be strictly less than the number of qubits used by the original circuit. If this were not the case then the circuit could equally well be run within the largest module, using the link qubits as computation qubits.
\end{itemize}
In this section we firstly demonstrate that the methods introduced in \cref{sec:solutions} do not produce distributed circuits requiring excessively large link qubits. Secondly we introduce an approach to limiting the size of the link qubit register. 

We explore the size of link qubit registers across a collection of distributed circuits. As the distributable packets are larger and longer lasting in the case of Pauli Gadget circuits, we will consider them here. As there was no noticeable difference in performance between networks in \cref{sec:benchmarks results heterogeneous}, we will consider only Small World networks. We take networks with an average module size of 4, considering both unbounded link qubit registers, and link qubit registers bounded to contain 3 qubits. We consider networks with 3, 4, and 5 modules, taking 3 networks of each size. The Pauli Gadget circuits are of the same size as the total number of computational qubits in the network, with 5 random circuits generated for each network. We use the \cesd{} distributor as it was found to be the best performing in the results of \cref{sec:benchmarks results heterogeneous}. The relevant results are presented in \cref{fig:limited links}.

\begin{figure}
    \centering
    \begin{subfigure}[b]{\columnwidth}
        \includegraphics[width=\columnwidth]{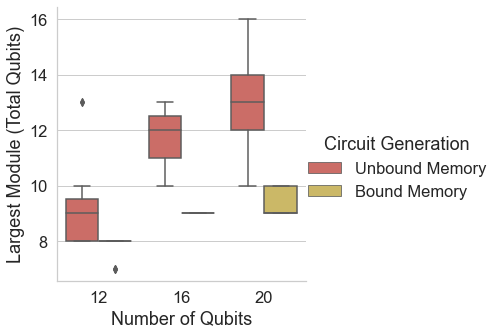}
        \caption{\textbf{Largest link qubits register size}. Largest Module is the size of the largest module in the network, measured as the sum of the sizes of the computational and link qubit registers. Link qubit memory is fixed to 3 qubits in every module; variance in module size is due to difference in computation memory size which is 4 on average.}
        \label{fig:limited links size}
    \end{subfigure}

    \vspace{5pt}
    
    \begin{subfigure}[b]{\columnwidth}
        \includegraphics[width=\columnwidth]{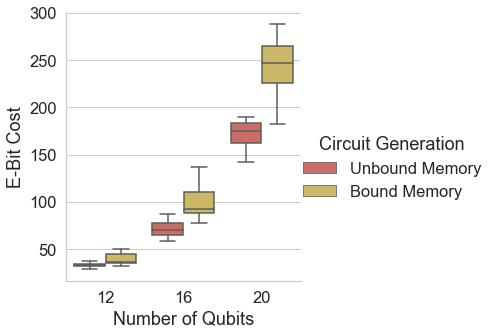}
        \caption{\textbf{Ebit cost.}}
        \label{fig:limited links cost}
    \end{subfigure}
    \caption{\textbf{Fixed link qubits register size.} Number of Qubits gives the size of the Pauli Gadget circuit, and therefore the total number of computational qubits in the network. Link qubit memory is fixed to 3 qubits in every module. In each plot boxes show the quartiles of the dataset, whiskers stretch to largest and smallest value within 1.5 of the interquartile range, and remaining points are outliers.}
    \label{fig:limited links}
\end{figure}

\paragraph{Link Qubit Register Size.} 

As seen in \cref{fig:limited links size}, and with few exception, the largest module required by circuits distributed onto networks with unbounded memory is smaller than the total number of qubits in the original circuit. This is more consistently the case as the number of modules in the network increases, and the size of the largest module appears to plateau. This suggests that the size of the link qubit registers is correlated with the average module size, rather than the number of qubits in the circuit.

\paragraph{Bounded Link Qubit Registers.} 

Our tool, \pytketdqc{}, checks whether the distributed circuit exceeds the link qubit register capacity of any module. If it does, the user may request \pytketdqc{} to amend it, at the cost of extra ebits. We now sketch the approach implemented in \pytketdqc{} to do so. As discussed in \cref{sec:to_pytket_circuit}, the generation of the distributed circuit proceeds iteratively, distributing each of the hyperedges of the \texttt{Distribution} one at a time. As we do so, we keep track of the available space of the link qubit registers of each module. Whenever the realisation of a hyperedge would cause the capacity to be exceeded, we make note of the offending module \texttt{A} and the non-local gate $g$ that this happened at. Then, we find the subset of hyperedges that share their qubit with module \texttt{A} (\ie{} those that have some of its gate-vertices allocated to \texttt{A}) and whose distributable packets span over gate $g$ --- not necessarily distributing it. These are the hyperedges that require the existence of a link qubit in module \texttt{A} at the time $g$ is distributed. If we split any of these hyperedges into two different ones --- by separating the gate-vertices that come before $g$ from those that come after --- we may remove the need to store its link qubit at the time of the bound violation. Thus, we simply need to use some heuristic to pick one of these hyperedges, update the \texttt{Distribution} splitting it, and run the circuit generation routine of \cref{sec:to_pytket_circuit} again; this process is repeated as many times as necessary to satisfy the user's bound to the link qubit registers. 

The heuristic we use to choose which of the hyperedges to split is simple: we pick the one whose gate-vertices immediately before and after $g$ are furthest apart in the circuit. Intuitively, this identifies the hyperedge whose link qubit in module \texttt{A} is the most `idle' at the time of the bound violation. Generally speaking, any circuit may be distributed using modules whose link qubit registers are only capable of storing a single link qubit (unless detached gates are used, in which case a minimum of two link qubits per module are required); the harsher the bound, the more ebits will be required to distribute the circuit.

The most relevant comparable result is that of \cite{Sundaram2022GeneralDistribution}, where a technique to bound link register size is introduced. Unlike ours, their main optimisation procedure already considers the bound to link qubit registers and, hence, their distributions tend to satisfy the bound more often than ours. However, as the authors explained, bound satisfaction is not guaranteed by their approach either, which means they sometimes need to apply a final pass similar to ours at the end of the optimisation. When comparing their pass with ours, we find their approach to be too strict: it picks one of the distributable packets causing the bound violation and opt to distribute each of its \CZ{} gates separately, consuming one ebit for each. In contrast, our approach amends the distributed circuit with less ebit overhead, at the cost of requiring a non-trivial search and repeat-until-success approach, which will take longer to run.

\cref{fig:limited links size} shows that strictly capping the size of the link qubit register to 3 limits the size of the largest module below that produced by distributing onto networks with unbounded link qubit registers. As expected, and as seen in \cref{fig:limited links cost}, the cost in ebits of the resulting distribution is increased in the case of bounded memory.

\end{document}